\documentclass[11pt]{article}

\usepackage{microtype}
\usepackage{graphicx}
\usepackage{subfigure} 
\usepackage{booktabs} % for professional tables
\usepackage{natbib}
\usepackage{epsfig,epsf,fancybox}
\usepackage{amsmath}
\usepackage{mathrsfs}
\usepackage{amssymb}
\usepackage{color}
\usepackage{multirow}
\usepackage{paralist}
\usepackage{verbatim}
\usepackage{algorithm}
\usepackage{algorithmic}
\usepackage{galois}
\usepackage{boxedminipage}
\usepackage{accents}
\usepackage{stmaryrd}
\usepackage[bottom]{footmisc}
\usepackage{microtype}
\usepackage{graphicx}
\usepackage{subfigure}
\usepackage{booktabs} % for professional tables
\usepackage{amsmath}
\usepackage{amssymb}
\usepackage{amsthm}
\usepackage[mathscr]{euscript}
\usepackage{nicefrac}
\usepackage{xspace}
\usepackage{wrapfig}

\usepackage{natbib}
\usepackage{tcolorbox}
\usepackage{url}
\usepackage{listings}
\usepackage[colorlinks, linkcolor={blue!70!black}, citecolor={blue!70!black}, urlcolor={blue!70!black}]{hyperref}
\usepackage{graphicx}
\usepackage{subfigure}
\usepackage{multirow}
\usepackage{xcolor}
\usepackage{algorithm}
\usepackage{algorithmic}

\newtheorem{theorem}{Theorem}[section]

\newtheorem{lemma}[theorem]{Lemma}

\newtheorem{definition}[theorem]{Definition}

\newtheorem{remark}[theorem]{Remark}

\usepackage{mathtools}
\usepackage{multirow}
\usepackage[flushleft]{threeparttable}
\usepackage{varwidth}
\usepackage{xcolor}
\usepackage{xfp, siunitx}
\usepackage{tikz}
\usepackage{bm}
\usepackage{tcolorbox}
\usepackage[mathscr]{euscript}
\usepackage{thmtools} 
\usepackage{thm-restate}

\definecolor{tab10_blue}{rgb}{0.121, 0.466, 0.705}
\definecolor{tab10_orange}{rgb}{1.0,   0.498, 0.054}
\definecolor{tab10_green}{rgb}{0.172, 0.627, 0.172}
\definecolor{tab10_red}{rgb}{0.839, 0.152, 0.156}
\definecolor{tab10_purple}{rgb}{0.580, 0.403, 0.741}
\definecolor{tab10_brown}{rgb}{0.549, 0.337, 0.294}
\definecolor{tab10_pink}{rgb}{0.890, 0.466, 0.760}
\definecolor{tab10_gray}{rgb}{0.498, 0.498, 0.498}
\definecolor{tab10_olive}{rgb}{0.737, 0.741, 0.133}
\definecolor{tab10_cyan}{rgb}{0.090, 0.745, 0.811}
\definecolor{Gray}{gray}{0.85}
\usepackage{enumitem}

\newcommand{\speedup}[1]{{\color{gray}(\ifdim #1 pt > 0.3pt #1\else $< #1$\fi{}$\times$)}}

\usepackage{url}

\definecolor{tab10_blue}{rgb}{0.121, 0.466, 0.705}
\definecolor{tab10_orange}{rgb}{1.0,   0.498, 0.054}
\definecolor{tab10_green}{rgb}{0.172, 0.627, 0.172}
\definecolor{tab10_red}{rgb}{0.839, 0.152, 0.156}
\definecolor{tab10_purple}{rgb}{0.580, 0.403, 0.741}
\definecolor{tab10_brown}{rgb}{0.549, 0.337, 0.294}
\definecolor{tab10_pink}{rgb}{0.890, 0.466, 0.760}
\definecolor{tab10_gray}{rgb}{0.498, 0.498, 0.498}
\definecolor{tab10_olive}{rgb}{0.737, 0.741, 0.133}
\definecolor{tab10_cyan}{rgb}{0.090, 0.745, 0.811}

\DeclarePairedDelimiterX{\lin}[2]{\langle}{\rangle}{#1, #2}

\DeclarePairedDelimiterX{\abs}[1]{\lvert}{\rvert}{#1}
\DeclarePairedDelimiterX{\norm}[1]{\lVert}{\rVert}{#1}
\DeclarePairedDelimiterX{\cbr}[1]{\{}{\}}{#1} % curly bracket
\DeclarePairedDelimiterX{\rbr}[1]{(}{)}{#1} % round bracket
\DeclarePairedDelimiterX{\sbr}[1]{[}{]}{#1} % square bracket

\renewcommand{\epsilon}{\varepsilon}

\providecommand{\R}{\mathbb{R}} % Reals
\DeclareMathOperator{\E}{\mathbb{E}}

\DeclareMathOperator{\sgn}{sign}
\makeatletter
\def\sign{\@ifnextchar*{\@sgnargscaled}{\@ifnextchar[{\sgnargscaleas}{\@ifnextchar{\bgroup}{\@sgnarg}{\sgn} }}}
\def\@sgnarg#1{\sgn\rbr{#1}}
\def\@sgnargscaled#1{\sgn\rbr*{#1}}
\def\@sgnargscaleas[#1]#2{\sgn\rbr[#1]{#2}}
\makeatother
\DeclareMathOperator*{\argmin}{arg\,min}
\DeclareMathOperator*{\argmax}{arg\,max}

\providecommand{\vc}{\mathbf{c}}

\providecommand{\vm}{\mathbf{m}}

\providecommand{\vp}{\mathbf{p}}

\providecommand{\cC}{\mathcal{C}}
\providecommand{\cD}{\mathcal{D}}

\providecommand{\cH}{\mathcal{H}}
\providecommand{\cI}{\mathcal{I}}

\providecommand{\cM}{\mathcal{M}}

\newtheorem{example}[theorem]{Example}

\newcommand{\e}{\varepsilon}

\makeatletter
\newcommand{\myitem}[1]{%
\item[\textbf{(#1)}]\protected@edef\@currentlabel{#1}%
}
\makeatother

\newcommand{\eg}{e.g.\xspace}
\newcommand{\ie}{i.e.\xspace}

\hypersetup{
    allbordercolors = [rgb]{0.090, 0.745, 0.811}, % tab10_cyan
}

\newcommand{\low}[1]{\underaccent{\bar}{#1}\xspace}
\newcommand{\high}[1]{\bar{#1}\xspace}

\usepackage
[a4paper,
left=2.5cm,
right=2.5cm,
top=3cm,
bottom=4cm]
{geometry}

\begin{document}

%%%%%%% TITLE PAGE %%%%%%%%%%%%%%%%%%%%%%%%%%%%%%%%%%%%%%%%%%%%%%%%%%%

\begin{center}

%{\bf{\LARGE{Data-Maximizing Mechanisms for Federated Learning}}}
{\bf{\Large{Mechanisms that Incentivize Data Sharing in Federated Learning}}}

\vspace*{.25in}
{\large{
\begin{tabular}{c}
Sai Praneeth Karimireddy$^{\diamond}$\footnotemark \quad Wenshuo Guo$^{\diamond 1}$ \quad Michael I. Jordan$^{\diamond, \dagger}$\\
\end{tabular}
}}

\vspace*{.15in}
\begin{tabular}{c}
Department of Electrical Engineering and Computer Sciences$^\diamond$ \\
Department of Statistics$^\dagger$ \\ 
University of California, Berkeley
\end{tabular}
\footnotetext{Equal contribution.}
\vspace*{.1in}

\begin{abstract}
\noindent
    Federated learning is typically considered a beneficial technology which allows multiple agents to collaborate with each other, improve the accuracy of their models, and solve problems which are otherwise too data-intensive / expensive to be solved individually. However, under the expectation other agents will share their data, rational agents may be tempted to engage in detrimental behavior such as \emph{free-riding} where they contribute no data but still enjoy an improved model.
    In this work, we propose a framework to analyze the behavior of such rational data generators. We first show how a naive scheme leads to catastrophic levels of free-riding where the benefits of data sharing are completely eroded. Then, using ideas from contract theory, we introduce \emph{accuracy shaping} based mechanisms to maximize the amount of data generated by each agent. These provably prevent free-riding without needing any payment mechanism.
    % We corroborate our theoretical results with numerical simulations.
\end{abstract}

\end{center}

% !TEX root = paper.tex

% Introduction %%%%%%%%%%%%%%%%%%%%%%%%%%%%%%%%%%%%%%%%%%%%%%%%%%%%%%%%%%%% 

\section{Introduction}
Data is a \emph{non-rivalrous good}---once produced, it can be repeatedly used multiple times without exhaustion. Thus, multiple firms can simultaneously use the data produced by any individual firm, increasing societal utility/welfare~\citep{jones2020nonrivalry}. To promote such multiple usage, data portability requirements have been widely legislated, \eg, the GDPR in the EU, CCPA in California, etc~\citep{mancini2021data}. As a consequence, services are required to enable a user to download any personal data collected and potentially re-upload it to a different service. These desiderata form a solid economic and legal basis for federated learning---a new paradigm in machine learning wherein multiple data-generating agents collaborate with each other to train a model on their \emph{combined} data so that all the agents end up with a better model than they would have obtained on their own~\citep{kairouz2021advances}. Such collaborative data sharing is already common in genomics research~\citep{weinstein2013cancer}, internet advertisement targeting~\citep{google_adsdatahub}, and is also gaining traction between networks of hospitals~\citep[see, e.g.,][]{sheller2018multi,wen2019federated,rieke2020future,flores2021federated}.

It is clear that once a certain amount of data has been produced, privacy issues aside, societal welfare is maximized by allowing free access to the data thereby making it a public good. However, under such an expectation, a rational agent may be tempted to \emph{free-ride}, \ie, consume the benefits of the data production by others without contributing any data themselves. This may lead to a collapse in the data generation with everyone wanting to free-ride. Such a problem inevitably arises with any public good~\citep{baumol2004welfare}. Further, even if no agent actually free-rides and everyone intends to contribute data out of altruism, the mere perception that others may be free-riding reduces pro-social behavior and willingness to contribute~\citep{choi2019contributors}. Thus, the long-term success of federated learning in particular and data portability in general critically require overcoming free-riding.

% Thus, preventing free-riding can lead to a virtuous circle of collaboration where agents may be willing to contribute more than they would otherwise. 

The overall consequence of free-riding to the system is that a lesser amount of data may be generated. Thus, we can equivalently formulate the problem of preventing free-riding as that of maximizing the amount of data generated by the agents. This motivates our main question:
\begin{quote}
    \emph{How do we design a system which incentivizes rational agents to contribute their fair share of data, thereby maximizing the accuracy of the resulting model and improving collaboration?
    }
\end{quote}
Maximizing the amount of data generated will arguably lead to the greatest long-term societal welfare, even if allowing free access to it gives better short-term welfare. 
% there are numerous other non-rivalous (\ie no/low cost for repeated simultaneous consumption) and excludable (\ie access can be restricted)
% Specifically, we will work with the following setting. 
% \begin{enumerate}
%     \item There is a central platform 
% \end{enumerate}
\paragraph{Motivating example.} Autonomous driving is data-intensive, with expensive data generation. Each data point involves a person physically driving. On the other hand, more data and more accurate models will potentially lead to reduced accidents and save lives. The high cost of data collection also means that only a select few providers will be able to raise capital required to collect enough data, limiting innovation.

Given that, a government agency may pass legislation requiring autonomous driving providers to share their data openly with each other with the following two goals: (i) hoping that each provider would now have access to a larger pool of data and can train a more accurate model, and (ii) forcing collaboration between the providers in order to encourage solving more ambitious problems. However, providers may react to such a data-sharing regulation by reducing the amount of data they collect and instead free-riding, defeating the purpose of the legislation. \emph{How should the government agency formulate its data-sharing legislation in order to maximize the total amount of data collected and shared?}

\subsection{{Contribution and summary of results}}

% In order to answer such a question, we first need a framework to model the behavior of data generating agents. Borrowing ideas from contract theory, we formulate a principal-agent model~\citep{laffont2009theory} where each agent has a cost associated with generating a data point and wants to improve the value of a model while minimizing said costs (Sec.~\ref{sec:framework}). Using this framework we show how giving unconditional access to the combined data (as is standard in federated learning) leads to catastrophic free-riding (Sec.~\ref{sec:free_ride}). Thus, overcoming free-riding requires conditioning the value of the model returned to the user upon their contribution (a technique we call \emph{value shaping}). If we know the agent's cost of data generation, we derive the optimal value curve which overcomes free-riding and leads to maximal data generation (Sec.~\ref{sec:value_shaping}). Finally, if the costs of an agent are unknown, we show in Sec.~\ref{sec:information_rent} how to design truth-revealing accuracy curves at the cost of some data to incentivize the agent to report their true cost.

% \wg{prefer to make contributions into a list?}
% \textcolor{tab10_brown}{Alt (see below)}
In this work, we formally introduce the \emph{data maximization incentivization problem} in federated learning, and design new mechanisms to achieve this goal. 
% Our framework provides a formal characterization of the tradeoff between incentivization and efficiency, and illustrates an ``information rent" that the platform has to suffer under information asymmetry.  The notion of information rent is a key concept in the economic literature on contract theory. 
In more detail:
\begin{itemize}
    \item  We formulate a principal-agent model~\citep{laffont2009theory} where each agent has a cost associated with generating a data point and wants to improve the accuracy of a model while minimizing said costs (Sec.~\ref{sec:framework}). Our formulation borrows ideas from contract theory while introducing new concepts that are specific to the federated learning setting.
    \item Using this framework we show how giving unconditional benefit of the combined data to all agents (as is standard in federated learning) leads to catastrophic free-riding where almost none of the agents contribute any data (Sec.~\ref{sec:free_ride}) at their optimal responses.
    \item Accordingly, we propose to tune the accuracy of the model received by an agent to their contribution. In the full-information setting when the agent's cost of data generation is known, we derive an optimal mechanism which overcomes free-riding and leads to maximal collaboration and data generation (Sec.~\ref{sec:accuracy_shaping}). 
    \item Finally, if the costs of an agent are unknown, we show (in App.~\ref{sec:information_rent}) how to design truth-revealing accuracy curves at some cost to the principal (\ie, information rent) to incentivize the agents to report their true costs.
\end{itemize}
Our framework can capture free-riding and the need for collaboration when faced with challenging learning problems. The latter is novel to our framework---we show that if the learning task is too challenging, then it is not economically viable for any single agent to tackle the problem. However, using incentivizing data-sharing mechanisms, it may be possible to share the costs among participants and solve it collaboratively.

\subsection{Related Work}

% As alluded to in the introduction, 
% The literature on mechanism design and federated learning is vast. We discuss the most closely related work in three verticals and include a detailed review of the broader literature in Appendix~\ref{app:related_work}.

\paragraph{Free-riding and fairness in federated learning.}
% The standard federated learning scheme aims to allow the agents to train collaboratively and leverage the benefit from other agents' data. However, such a scheme may incentive strategic agents to contribute less data in order to minimize their data collection cost and maximize the gain from participating in the federated learning mechanism. 
Recent work has explored such free-riding behavior in federated learning schemes, with various incentive models proposed~\citep{sarikaya2019motivating,lin2019free,ding2020incentive, fraboni2021free}. Most of this work has, however, focused on a taxonomy of free-rider attacks or the detection of such attacks, instead of proposing mechanisms that incentivize maximal data contribution. In this work, we consider a mechanism for data sharing under the standard federated learning setting such that rational agents are incentivized to contribute their maximal amount of data.

\paragraph{Contract theory in federated learning.} 
Preventing free-riding behavior in federated learning is notoriously hard because the data collection and costs are private to the agents. This information asymmetry and the existence of a central server (a ``principal") suggests connections to contract theory, which studies the design of incentive mechanisms under a principal-agent model, where the agent possesses private information about their costs. Recently, there has been an emerging line of research exploring the application of contract theory for federated learning~\citep{kang2019toward,kang2019incentivea,kang2019incentiveb, cong2020game, zhan2020learning,lim2021towards, tian2021contract}. In particular,~\citet{tian2021contract} proposed a contract-based aggregator under a multi-dimensional contract model over two possible types of agents and showed improved model generalization accuracy under that contract. However, their mechanism focused on eliciting the private type information instead of maximizing the data contribution. To the best of our knowledge, our work is the first to use contract theory for \textit{data maximization} in federated learning. Further, prior work has focused on how to design payments to agents, rather than the accuracy-shaping problem that we focus on here without any payment usage.

\paragraph{Mechanism design for collaborative machine learning.} More broadly, this work is related to an active line of research on mechanism design for collaborative machine learning, which involves multiple parties each with their own data, jointly training a model or making predictions in a common learning task~\citep{sim2020collaborative, xu2021gradient}. In collaborative machine learning, a major focus has been the design of model rewards (i.e., data valuation) in order to ensure certain fairness or accuracy objectives. Towards that goal, there has been model rewards proposed based on notions from the cooperative game theory literature such as the Shapley value~\citep{jia2019towards, wang2020principled}. However, the guarantees of these model rewards depend on the assumption that the agents are already willing to contribute the data they have. In this work, we study a different incentivization task for data maximization.

Lastly, in this work, we have focused on the data-maximization goal under individual rationality and accuracy-shaping. More broadly, there are other objectives which are of interest in federated learning, such as fairness and welfare objectives, that have been under active study~\citep{donahue2020model, donahue2021optimality}. 
% Under our setting when the data is exchangeable, we show that it is impossible to maximize welfare assuming that the only intervention is accuracy shaping, which implies a trade-off between data maximization and welfare maximization. However, such a trade-off may disappear when the data is no longer exchangeable, e.g. each agent only holds data that comes from a biased local distribution. A thorough analysis over these different objectives and their trade-offs is an interesting open question for future research. \spk{I am not sure a discussion of welfare is super relevant. Further, I am not sure we will get around to formally proving ``it is impossible to maximize welfare assuming that the only intervention is accuracy shaping''}

% \spk{We should also probably discuss connections to revenue maximization in auctions.} \wg{currently have not added that, not very sure where it fits} \spk{removed comment - I guess not super relevant.}

% \spk{Seems to be a lot of work using contract theory in FL! Added a few references I found.}

% Modeling %%%%%%%%%%%%%%%%%%%%%%%%%%%%%%%%%%%%%%%%%%%%%%%%%%%%%%%%%%%% 

\section{Modeling an Individual Agent}\label{sec:framework}

We begin with modelling the learning task and objective for an individual agent. We then provide a characterization of the optimal data contribution for each single agent without participating in a federated learning scheme. 

\subsection{Learning problem}

There are $n$ agents all of whom want to solve a \emph{common} learning problem. This is often true in federated learning since coalitions form around solving some particular task. Concretely, we assume that all agents want to maximize an accuracy function, $a(\cD): 2^\cD \rightarrow [0, 1]$, for a dataset $\cD$.
We also assume each of them has access to the same data distribution and that we are in an i.i.d.\ setting. This holds true if the data is generated by manually labelling a subset of an already public unlabelled dataset, as is common in machine learning; \eg, Cifar~\citep{krizhevsky2009cifar} and ImageNet~\citep{imagenet2015}. This assumption is arguably also valid in our autonomous driving example where each data point involves a random path taken under random external conditions. %The i.i.d.\ assumption can be weakened to the assumption that the data is \emph{exchangeable}. Hence, we can simplify the accuracy function $a(\cdot)$ to depend only on the \emph{size} of the dataset $m=\abs{\cD}$. For convenience, we will treat dataset sizes as a continuous real even though it is a nonnegative integer.
Thus, every agent wants to maximize 
\begin{equation}\label{eqn:acc}
a(m): \R_{\geq 0} \rightarrow [0, 1] = \max(0, b(m))\,, \ \text{where $b$ is continuous, non-decreasing and concave.}
\end{equation}
We also assume w.l.o.g that $a(0) = 0$ and $\lim_{m \rightarrow \infty}a(m) > 0$. We will first show how an agent concerned with obtaining high test accuracy on a learning problem can be modeled by our formulation.
are motivated by empirical observations~\citep{kaplan2020scaling}, and standard generalization guarantees.
\begin{example}[Generalization bounds]\label{example:ERM}
Suppose we want to learn a model $h$ from a hypothesis class $\cH$ which minimizes the error over data distribution $\cD$, defined to be $R(h) := \E_{(x,y) \sim \cD}[e(h(x), y)]$, for some error function $e(\cdot) \in [0, 1]$. Let such an optimal model have error $(1 - a_{opt}) < 1$. Now, given access to $\{(x_l, y_l)\}_{l \in [m]}$ which are $m$ i.i.d.\ samples from $\cD$, we can compute the empirical risk minimizer (ERM) as $\hat h_m = \argmin_{h \in \cH} \sum_{l \in [m]} e(l(x), y)$. Then, standard generalization bounds \cite[Theorem~11.8]{mohri2018foundations} imply that with probability at least $99\%$ over the sampling of the data, the accuracy is at least%\vspace{-1.2em}
\begin{equation}\label{eqn:generalization}
    1 - R(\hat h_m) \geq \cbr[\Big]{b(m) := a_{opt} - \frac{\sqrt{2k(2 + \log(m/k))} + 4}{\sqrt{m}}}\,.
\end{equation}
Here $k \geq 1$ is the pseudo-dimension of the set of functions $\{(x,y) \rightarrow e(h(x),y) : h \in \cH\}$, which is a measure of the \emph{difficulty} of the learning task. We can then define our accuracy function $a(m) = \max(0, b(m))$. Note that $a(m) \in [0,1]$ and $\lim_{m \rightarrow \infty} = a_{opt} > 0$. Further, $b(m) \leq 0$ for $m \leq 1$ and $b(\cdot)$ is concave for any $m \geq 1$, satisfying \eqref{eqn:acc}. While \eqref{eqn:generalization} does satisfy our assumptions and can be used to understand our framework, it is too unwieldy to perform exact analysis. So we will use \eqref{eqn:generalization} for simulations, but use a simplified expression for our analytic analysis:%\vspace{-0.5em}
\begin{equation}\label{eqn:simple-generalization}
    b(m) := a_{opt} - 2\sqrt{k/{m}}\,.
\end{equation}
\end{example}
% \textcolor{red}{We will also assume that $\lim_{m \rightarrow \infty}a(m) < c$ for some constant $c$ so that the utility is bounded and cannot increase to infinity.}
%\vspace{-1em}
We use the above formulation as a running example in the rest of the paper. The next example shows a more stylized, but perhaps more realistic setting where a company derives some abstract value from data but only after a cutoff due to regulation.

\begin{example}[Starting costs and minimum viability]
Consider an autonomous driving provider who is training large machine learning models. The dependence of the final accuracy $b(\cdot)$ of such a model typically scales as a power law in the number of training data points; i.e., $b(m) = 1 - \nicefrac{\beta}{m^\alpha}$ for some $\beta > 0$ and $\alpha \in (0,1]$~\citep{kaplan2020scaling}. For such a provider, improved accuracy in their models might directly imply better customer satisfaction and hence more sales. It might also imply lower downstream costs in term of law-suits, etc. However, a small number (say ten) data points are not in themselves of any value. This is because the true value needs to overcome additional fixed (non data-related) costs, as well as have sufficient accuracy to pass safety regulation. Thus, in practice, the real value of data is better captured by $a(m) = \max(0 , 1 - \nicefrac{\beta}{m^\alpha} - \tau)$ for accuracy threshold $\tau$. Thus, there is a minimum viable dataset $m^0 \geq \rbr*{\nicefrac{\beta}{1 - \tau}}^{\nicefrac{1}{\alpha}}$. The value of the data below this threshold is zero, and its value beyond is concave and non-decreasing.
\end{example}
\subsection{Agent's objective and optimal solution}

Each agent $i$ has a marginal fixed cost $c_i > 0$ for producing a data point. Their cost for producing a dataset $\cD$ with $m$ number of data points is then:%\vspace{-0.5em}
\begin{equation}
    cost_i(m) = c_i m. 
\end{equation}

\begin{figure*}[!t]
\centering
\begin{tabular}{cccc} 
\includegraphics[height=0.18\textwidth]{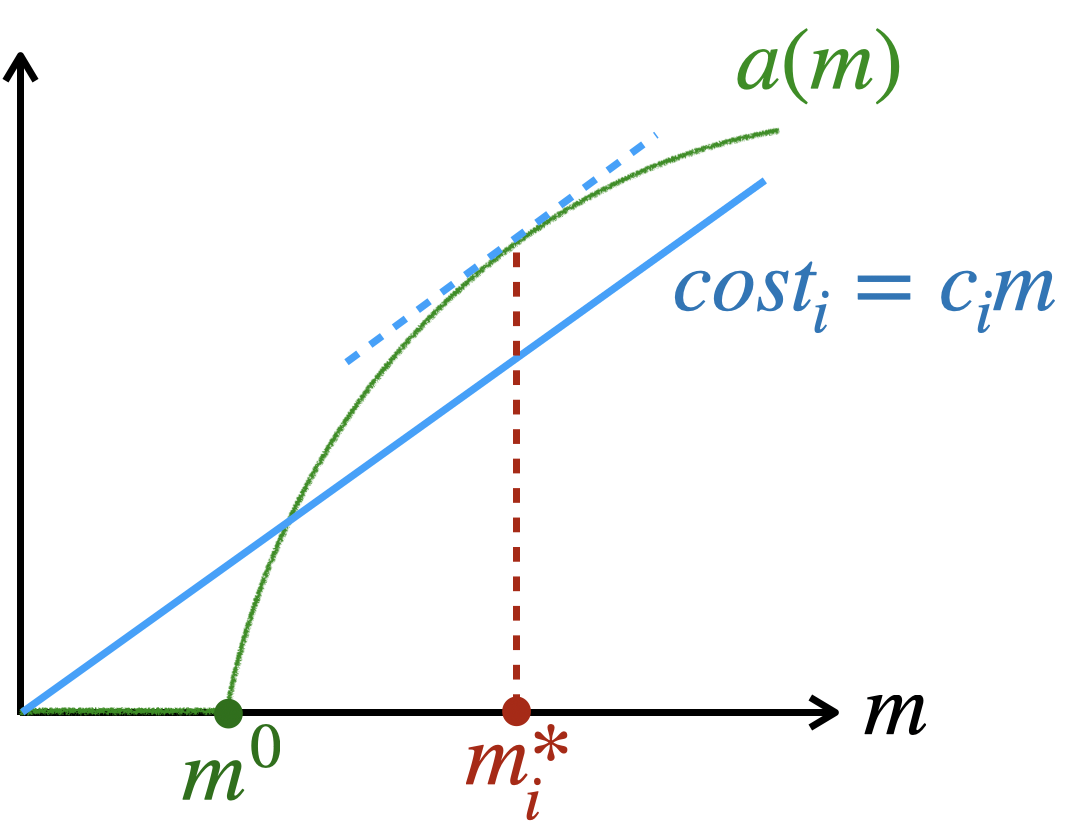} &
\includegraphics[height=0.17\textwidth]{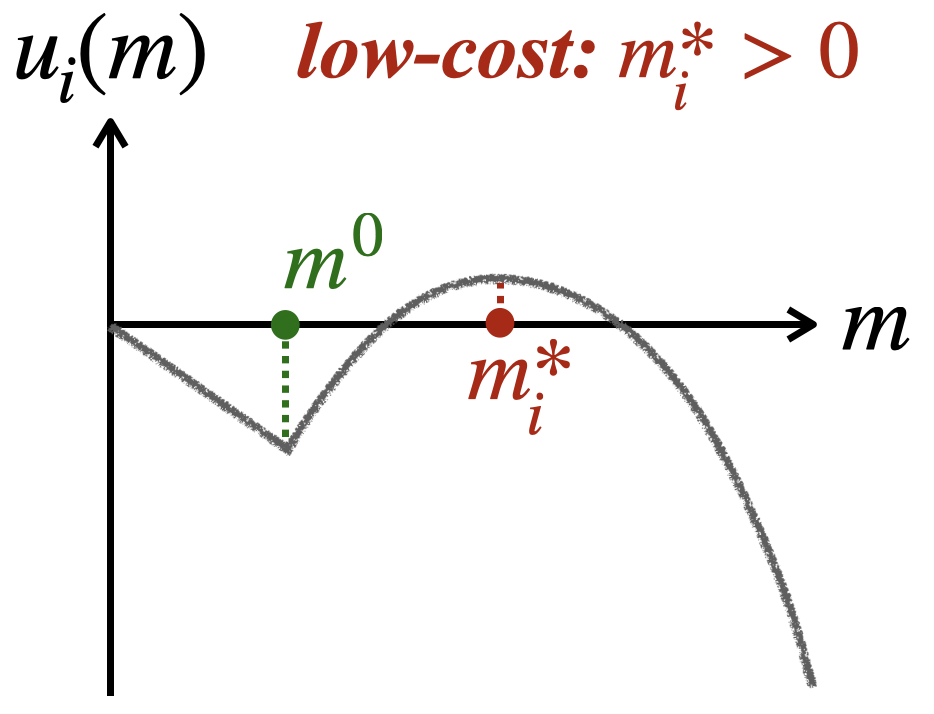} & 
\includegraphics[height=0.17\textwidth]{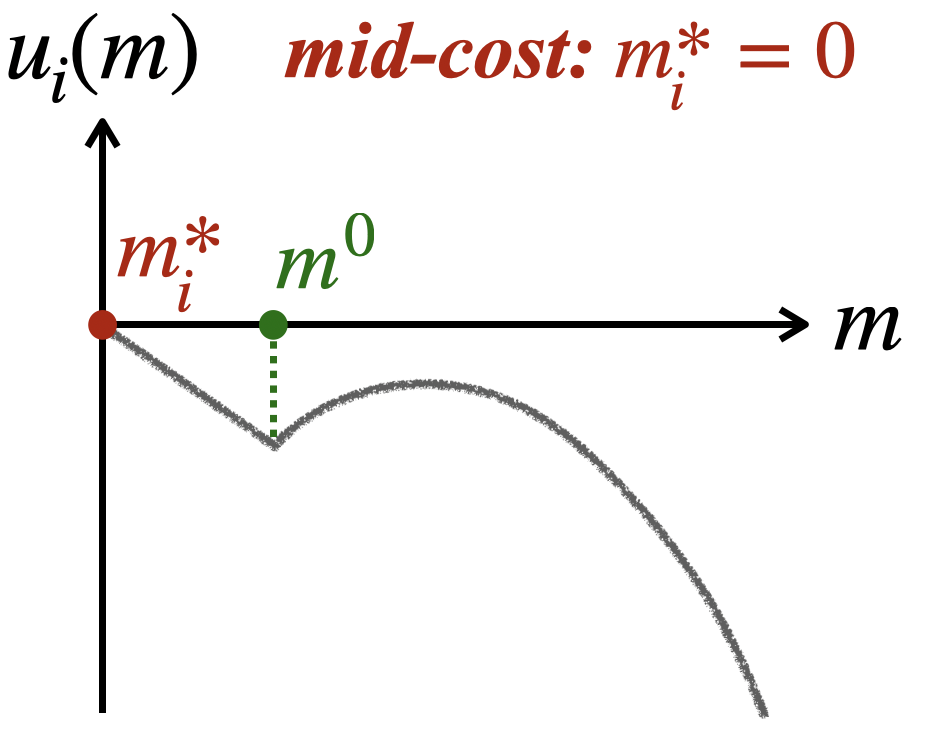} &
\includegraphics[height=0.17\textwidth]{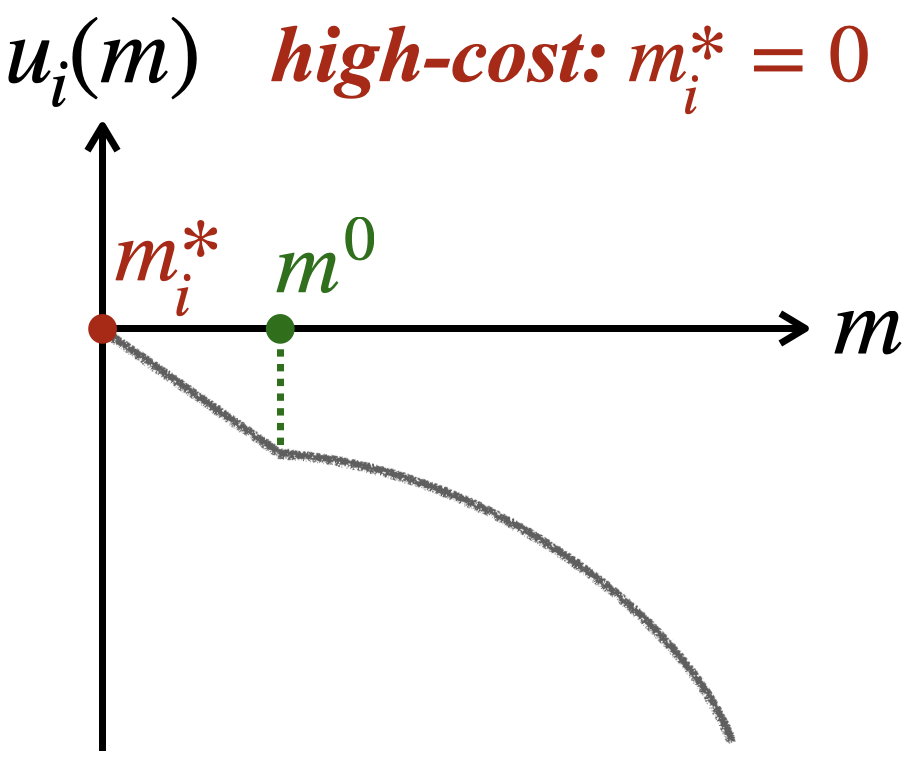}\\%\vspace{-0.5em}
\scriptsize{$(a)$} & \scriptsize{$(b)$} & \scriptsize{$(c)$} & \scriptsize{$(d)$}  \\
\end{tabular}%\vspace{-0.5em}
\caption{Illustration of the optimal amount of data for a single agent. \textit{(a):} Accuracy and cost versus the dataset size.\textit{(b)-(d):} Utility function of a low/mid/high-cost agent versus the dataset size.}\label{fig:single-agent} %\vspace{-1em}
\end{figure*}
%\vspace{-0.7em}
When manually labelling a dataset or when training an autonomous-driving model, this cost $c_i$ may represent the time spent by researchers/employees or an amount paid to crowd-sourced workers. The cost $c_im$ may also represent the risk associated with privacy loss for the agent for revealing $m$ of their data points. By incurring this cost, they can train a model with accuracy $a(m)$. We assume that the utility improves linearly with increasing accuracy $a(\cdot)$. For example, each accurate product recommendation may lead to a sale or correct digital ad placement may lead to a click and hence ad revenue. This is also true if each error represents costly consequences. Each error by a medical diagnostic model, a loan application evaluation model, or a autonomous driving model may lead to significant suffering. Thus, the utility of an agent is improve accuracy for the least cost; \ie, to maximize%\vspace{-0.5em}
\begin{equation}\label{eqn:utility}
u_i(m) = a(m) - c_i m.
\end{equation}

% Given this, we can predict the behavior of a rational agent working on their own. They will  generate $m_i^\ast$ which maximizes \eqref{eqn:utility}.

\begin{restatable}[Optimal individual generation]{theorem}{individualtheorem} \label{thm:ind}
Consider an individual agent $i$ with marginal cost per data point $c_i$ and accuracy function $a$ satisfying \eqref{eqn:acc} working on their own. Then, the optimal amount of data $m_i^\ast$ is:%\vspace{-1.3em}
\begin{equation}\label{eq:opt-single-agent}
 m_i^\ast =
    \begin{cases}
        0 &\text{if  ${\max_{m_i \geq 0} u_i(m_i)} \leq 0$};\\
        \alpha_i^\ast \text{, such that } b'(\alpha_i^\ast) = c_i &\text{otherwise.}%\vspace{-0.5em}
    \end{cases}
\end{equation}
Further, for agents $i,j $ with costs $c_i \leq c_j$, their utility satisfies $u_i(m_i^\ast) \geq u_j(m_j^\ast)$ and $m_i^\ast \geq m_j^\ast$.
\end{restatable}

As Figure~\ref{fig:single-agent} shows, if the learning problem is too hard ($m^0$ is large) or if the marginal cost $c_i$ is too high, the problem becomes infeasible for an individual agent to solve with $m_i^\ast = 0$. Such cases are especially important in federated learning where we want to enable agents to solve problems together which they cannot on their own. In other cases, the agent collects $m_i^\ast > 0$ data points.
% \todo{Add discussion data generated by flat-concave curve here (with figures of accuracy, and utility cases).} \wg{@SPK: this is addressed?} \spk{LGTM!}

\begin{wrapfigure}{r}{7cm}\label{fig:example-2}\vspace{-1.1em}
\centering
    \includegraphics[width=0.7\linewidth]{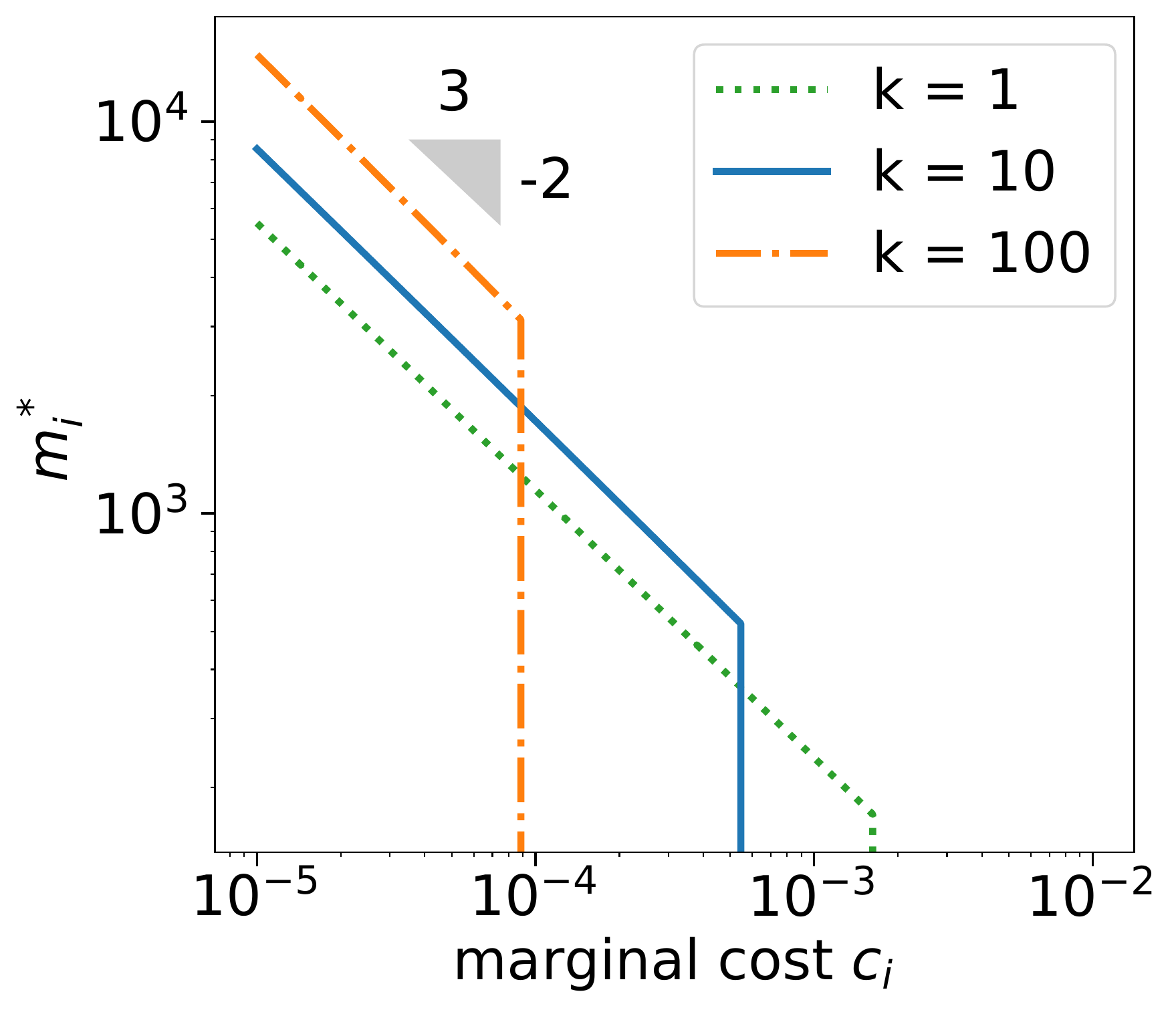}\vspace{-1.1em}
  \caption{Optimal data contribution $m_i^\ast$ versus the marginal cost $c_i$ for different number of total agents.}\vspace{-1.1em}
\end{wrapfigure}

% The class of utilities we consider is motivated by existing bounds for various learning problems. We now instantiate the class of utilities $u_i(m_i) = \max(0, b(m_i)) - c_i m_i$ with the standard generalization bounds for empirical risk minimization (ERM).
% Let us instantiate and illustrate our assumptions and results thus far with an example.

\begin{example}[Computing individual contributions]\label{example:ind-opt}
Consider the accuracy function arising from the generalization guarantees in Example~\ref{example:ERM}, and an agent with marginal cost $c_i$. Utilizing the simpler accuracy definition in \eqref{eqn:simple-generalization}, we can derive the following. The utility of agent $i$ becomes $u_i(m) = \max(0, a_{opt} - 2\sqrt{k/m}) - c_i m$. For this function, the minimum amount of data before accuracy starts to improve is $m^0 = 4a_{opt}^2/k$. The maximum cost at which the learning problem is viable is $c_i \leq a_{opt}^3/3k$. For any higher cost, the agent does not attempt to solve the problem with $m_i^\ast = 0$. When the cost is lower than this threshold, the agent generates data $m_i^\ast = k^{1/3} c_i^{-2/3}$ and obtains an utility of $u_i(m_i^\ast) = a_{opt} - 3 (kc_i)^{1/3}$. As expected, the minimum viability cost threshold scales inversely with the difficulty of the learning problem. Beyond this threshold, an agent's contribution decreases with cost and increases with the difficulty.
\end{example}%\vspace{-1em}

Empirically, we can plot the optimal individual generations using the full analytic form of the utility function (see Eq~\eqref{eqn:generalization}). Figure~\ref{fig:example-2} shows the optimal data contribution $m_i^\ast$ versus the marginal cost $c_i$ for different number of total agents on a log-log scale. We see that the optimal contribution decreases with cost as $m_i^\ast \propto c_i^{-2/3}$, matching the theory. The vertical lines indicate the cutoff for minimum viability---beyond this, the cost is too high for the problem to be solvable by an individual. This minimum viability cost is smaller for more harder problems (larger $k$), but the optimum contribution increases with increasing $k$ once this threshold is crossed.

% We will assume that for each agent, its utility is composed by two parts. One is a reward function of using the dataset for training proposes, such as accuracy. We assume that $a(\cdot)$ is increasing and concave (to model the diminishing return) with regard to the dataset size. 

% \begin{equation}
%     a(\cD): \cD \rightarrow [0, 1].
% \end{equation}

% The other is a cost function of collecting the data. We model this as a linear cost:
% \begin{equation}
%     cost_i(\cD) = c_i |\cD|.
% \end{equation}

% The utility of the agent is:
% \begin{equation}
% u_i(\cD_i) = a(\cD_i) - c_i |\cD|.
% \end{equation}

% Free-riding %%%%%%%%%%%%%%%%%%%%%%%%%%%%%%%%%%%%%%%%%%%%%%%%%%%%%%%%%%%% 

\section{Modeling Multiple Agents and Catastrophic Free-Riding} \label{sec:free_ride}
In this section we will study how agents behave when collaborating with each other as they do in federated learning. For this, we use a principal-\emph{multi}-agent framework where the server who sets up the federated learning server is the principal.%\vspace{-0.5em}

\subsection{Interaction between agents and server}%\vspace{-0.5em}
% The interaction proceeds in three steps: i) the server publishes a mechanism, ii) each agent generates and transmits some data to the server, and iii) the server aggregates the agents' data and transmits back a trained model to each agent following the mechanism. 
% Each of these steps in specified in more detail.
% Here, we ignore privacy constraints and assume that the agents transmit their raw data to the server (see Sec.~\ref{sec:discussion}). The server only transmits back 

The interaction between the federated learning server and the agents is formalized by a mechanism 
\begin{equation}
    \cM(\vm): \R_{\geq 0}^n \rightarrow [0, 1]^n\,, \text{ which maps agents' contributions to accuracies.}
\end{equation} We assume that each agent $i$ generates and transmits $m_i$ i.i.d.\ data points to the server, ignoring privacy concerns (see  Sec.~\ref{sec:discussion}). Based on these contributions, the mechanism assigns accuracies to the models received by the clients; \ie, if agent $i$ contributes $m_i$ data points it receives a model with accuracy $a_i \in [0,1]$, where $\cM(m_1, \dots, m_n) = (a_1, \dots, a_n)$. 
% The standard federated learning setup which returns a model trains on the combined dataset to everyone corresponds to the mechanism 
% \begin{equation}
% \cM(\vm) = \rbr*{a(\textstyle \sum_{j}m_j)\, \forall i \in [n]}\,.    
% \end{equation}

The interaction proceeds in three steps: (i) first the server publishes a mechanism $\cM$, then (ii) each agent generates and transmits some data $m_i$ to the server, and finally (iii) the server returns a trained model to each agent following the mechanism. 
Note that the agents decide how much data to generate adaptively \emph{after} knowing the mechanism $\cM$. However, they do not have any bargaining power---they cannot re-negotiate the mechanism---but can only decide if they join or not. We also disallow monetary compensation or exchanges between the parties since implementing them adds additional complexity. The only guarantee is that the server truthfully executes the protocol $\cM$. 
% The latter holds if either the server is trusted, or there is a trusted third-party (such as a government) which guarantees contracts are enforced.

Given that the server necessarily needs to follow through on the mechanism, we need to make sure the mechanism is implementable. %The server cannot offer more accuracy than what the combined data contributions allow.
% %\vspace{-0.2em}
\begin{definition}[\textbf{Feasible mechanism}]\label{def:feasible}
A mechanism which returns accuracy $[\cM(\vm)]_i$ to agent $i$ is said to be \emph{feasible} if for any $i \in [n]$ and any $\vm \in \R_{\geq 0}^n$, it satisfies $[\cM(\vm)]_i \leq a(\sum_j m_j)$\,.
\end{definition}%\vspace{-0.7em}
This is because we can pool together all the agent contributions $\vm$ and train a model to accuracy $a(\sum_j m_j)$. Since $a(\cdot)$ is monotone, this is an upper bound on the accuracy which can be obtained. However, it is always possible to use a subset of this data, or degrade the model in a controlled way using noisy perturbations. Thus, this captures mechanisms which are implementable in practice.

Faced with a potential feasible mechanism $\cM$, an agent has to decide whether to join or simply train on their own. A mechanism which offers an especially bad accuracy would discourage an agent and they would likely leave the server and train on their own. We will formalize this next.
\begin{definition}[\textbf{Individual rationality (IR)}]\label{def:ir}
Given data contributions $\vm$ by the $n$ agents with costs $\vc$, the mechanism provides a model with accuracy $[\cM(\vm)]_i$ to agent $i$. Such a mechanism $\cM$ is said to satisfy IR if for any agent $i \in [n]$ and any contribution $\vm$,%\vspace{-0.5em}
\begin{equation}\label{eqn:ir}
    [\cM(\vm)]_i - c_i m_i \geq a(m_i) - c_i m_i\,.
\end{equation}
\end{definition}
%\vspace{-0.5em}
A mechanism which satisfies individual rationality guarantees that for any agent the accuracy of the model received (and hence their utility) will be no worse than if they trained on their own. 
% If the inequality in \eqref{eqn:ir} holds strictly, we say $\cM$ is \emph{strictly} IR. 
Since IR guarantees that all rational agents will participate in our mechanism, and participation is key to success of a federated learning platform, we will restrict our focus henceforth to mechanisms which satisfy IR. In the context of our running example of data-sharing legislation for autonomous driving, ensuring the mechanism satisfies IR means that the legislation is a win-win for all parties encouraging support and compliance---consumers get safer models and car companies get higher utilities.

Given any mechanism $\cM$, we would like to argue about how rational agents would respond and how much data they would contribute. For this, we use the notion of an equilibrium.
\begin{restatable}[Existence of pure equilibrium]{theorem}{equlirbiumtheorem}\label{thm:equilibrium}
    Consider a feasible mechanism $\cM$ which can be expressed as:
    \[
        [\cM(m_i ; \vm_{-i} )]_i = \max\rbr*{0, \nu_i(m_i ; \vm_{-i})}\,,
    \]
    for a function $\nu_i(m_i; \vm_{-i})$ which is \emph{continuous} in $\vm$ and \emph{concave} in $m_i$. For any such $\cM$, there exists a pure Nash equilibrium in data contributions $\vm^{eq}(\cM)$ which for any agent $i$ satisfies,
\begin{equation}\label{eqn:nash}
    [\cM(\vm^{eq}(\cM))]_i - c_i m_i^\cM \geq [\cM( m_i, \vm^{eq}(\cM)_{-i})]_i - c_i m_i\,, \text{ for all } m_i \geq 0\,.
\end{equation}
% Further, the equilibrium contribution $[\vm^{eq}(\cM)]_i$ is non-increasing in the agent's cost $c_i$.
\end{restatable}
Thus, under reasonable conditions on the mechanism $\cM$ which are satisfied for all the mechanisms we consider, an equilibrium always exists. Note that the mechanism is not concave because of the presence of a $\max(0, \cdot)$, and the resulting utilities of the agents are not even quasi-concave. Despite this, our proof uses the specific properties of our setting to prove existence, and may be more broadly applicable to study non-concavities arising from ``minimum viability'' as is captured by the max. The existence of such an equilibrium allows us to confidently use the data contributions at this equilibrium to evaluate and compare different mechanisms.
%\vspace{-0.5em}
\subsection{Standard federated setting}%\vspace{-0.5em}
We now examine the behavior of rational agents in the standard federated learning. Returning a model trained on the combined dataset to everyone corresponds to the mechanism %\vspace{-0.5em}
\begin{equation}
\cM(\vm) = \rbr*{a(\textstyle \sum_{j}m_j)\,, \forall i \in [n]}\,.    
\end{equation}
Clearly, this mechanism is feasible (Def.~\ref{def:feasible}) and also satisfies individual rationality (Def.~\ref{def:ir}) since the accuracy function $a(\cdot)$ is non-decreasing and $\sum_j m_j \geq m_i$ for any $i \in [n]$. In fact, given a data contribution $\vm$, this mechanism maximizes utility for all agents. This observation may at first seem like a strong argument in favor of this standard scheme. However, recall that the agents choose their contribution $\vm$ \emph{after} the server publishes the mechanism $\cM$. Thus, we need to first analyze how much data rational agents would contribute.

\begin{restatable}[Catastrophic free-riding]{theorem}{federatedtheorem}\label{thm:federated}
    Let $\{m_i^\ast\}$ be the equilibrium contributions of agents when alone, with agent with least cost having a contribution $m_i^\ast =: m^{fed}$. The standard federated learning mechanism corresponding to $[\cM(\vm)]_i = a(\sum_{j}m_j)$ for all clients $i$ is feasible and IR. Further, the total amount of data collected at equilibrium remains $m^{fed}$, with only the lowest cost agent contributing:%\vspace{-1em}
    \begin{equation}
        m_i^{eq} = \begin{cases}
            m_{i}^\ast =: m^{fed} &\text{ if } i = \argmin_{j}c_j\\
            0 &\text{ otherwise.}%\vspace{-0.5em}
        \end{cases}
    \end{equation}
\end{restatable}
The agent with the least cost $c_{\min} = \min_j c_j$ would have collected $m^{fed}$ amount of data on their own. For any other agent $i$, $c_i \geq c_{\min}$ and so $m_i^\ast \leq m^{fed}$. Thus, agent $i$ would already have access to data sufficient to satisfy them. The increase in accuracy $a(\cdot)$ for collecting an additional data point beyond this is less than the marginal cost $c_i$ incurred. This results in catastrophic free-riding, with only a single agent collecting data.
In the case of our autonomous vehicle legislation example, this implies that naive data-sharing legislation will have no effect. The providers would strategically reduce the amount of data collected so that each of them is using almost the same amount of data as before. %For the providers, however, this outcome may not be bad. Even if the rest pay the one data-generating provider to further incentivize its data collection, they make massive savings in their cost of data collection. In other words, the entire surplus generated by the providers sharing data with each other gets captured by the providers themselves, leaving the accuracy of the final model (and hence the utility of the end-consumer) unchanged.

\begin{remark}[Failure of collaboration]
Consider the case where $m^{\ast}_i = 0$ for all agents $i$, either because the learning problem is too hard or because the cost of data collection is too high for any individual agent. Theorem~\ref{thm:federated} shows that $m^{fed} = 0$ and no data will be collected even with collaboration. Thus, if a problem is too costly to solve by an individual, it will remain insurmountable when collaborating with rational self-interested agents. This defeats one of the main motivations of federated learning.
\end{remark}

% \todo{Add discussion on consequences.
% Use this to motivate alternative formulations.}
% This equilibrium is not strictly IC for the player with the least cost since they would have identical utility on their own - they receive no added benefit from joining.

% Known costs %%%%%%%%%%%%%%%%%%%%%%%%%%%%%%%%%%%%%%%%%%%%%%%%%%%%%%%%%%%% 

\section{Accuracy Shaping with Verifiable Costs}\label{sec:accuracy_shaping}%\vspace{-0.5em}
How do we design mechanisms which prevent free-riding? In this section we will study this question assuming everyone (the server and the agents) know the costs $\vc$ involved in producing the data (we study the unknown costs setting in Section~\ref{sec:information_rent}), or at the least the costs can be verified; i.e., the agent cannot incur cost $c$ and report a different cost $\tilde c$. This is sometimes justifiable---the cost of driving a vehicle to generate a data point by an autonomous-driving provider can be easily estimated by all parties, as can the cost of labelling a data point by a crowd-worker. We formalize our goal of data maximization and give a simple optimal mechanism for it. Then, we look at some implications of the proposed solution---how fairly it distributes the surplus and potential moral hazards it might induce.%\vspace{-0.5em}
\subsection{Data maximization using accuracy shaping}%\vspace{-0.5em}
A mechanism $\cM$ is data-maximizing given costs $\vc$ if it maximizes the data collected at equilibrium.
\begin{definition}[\textbf{Data Maximization}]\label{def:data-max}
Suppose that given a mechanism $\cM$, let $\vm^{eq}(\cM)$ correspond to the amount of data generated by the agents at equilibrium. $\hat\cM$ is \emph{data-maximizing} if it maximizes the amount of data collected at equilibrium%%\vspace{-0.5em}
\begin{equation}\label{eqn:data-max}
   \hat\cM \in  \argmax_{\cM} \textstyle\sum_{j}[m^{eq}(\cM)]_j\,,  \text{subject to $\cM$ being feasible and satisfying IR.}
\end{equation}%\vspace{-1em}
\end{definition}
% If there are multiple $\vm^\cM$ which all satisfy \eqref{eqn:nash}, \eqref{eqn:data-max} utilizes the worst one (i.e., has the least data). 
% Agent $i$ will generate $m_i^\ast$ without requiring any additional incentive. However, after this point, the marginal improvement in the accuracy is less than the cost per extra data point $c_i$. Thus, in order to incentivize the agent to keep generating data beyond this point, we would need to provide them extra data to make it worth the marginal cost. Before stating the exact mechanism, let us define our goal.

\begin{wrapfigure}{r}{8cm}\vspace{-1.1em}
    \centering
    \vspace{-1.1em}
    \includegraphics[width=\linewidth]{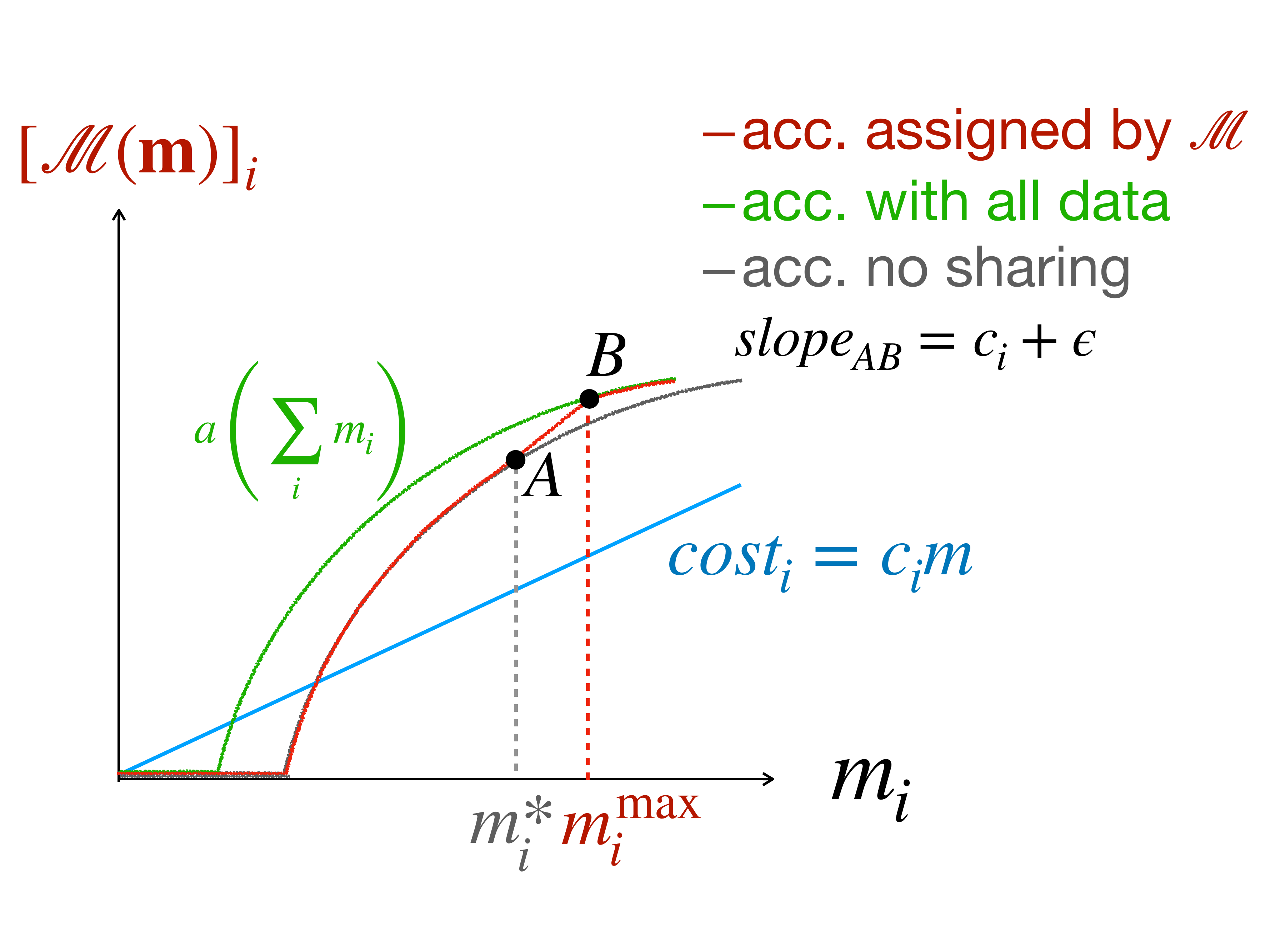}%\vspace{-0.8cm}
    \caption{Illustration of the accuracy shaping mechanism. \textit{(red curve)}: model accuracy returned to agent $i$ by the mechanism; \textit{(grey curve)}: model accuracy for agent $i$ without participation; \textit{(green curve)}: model accuracy if agent $i$ receives all the data from the other agents.}%\vspace{-0.5cm}
    \label{fig:accuracy-shaping-known-costs}
\end{wrapfigure}

Given a mechanism $\cM$, we first have to reason about how much each agent would contribute. As we will see, there is always an equilibrium  contribution $\vm^{eq}(\cM)$ satisfying \eqref{eqn:nash} such that no agent can improve their utility by unilaterally changing their contribution. If all players are rational (and such an equilibrium is unique), then such a point is a natural attractor with all the agents gravitating towards such contributions. So, a reasonable goal for us as a mechanism designer is to pick an $\cM$ which maximizes the amount of data collected when all players are contributing such equilibrium amounts.

If we give $\Delta m_i$ free data to agent $i$, then at equilibrium  they will reduce the data they generate---they will only generate $(m_i^\ast - \Delta m_i)$ additional data. To prevent this, our key insight is to condition the amount of extra data $\Delta m_i$ on the agent's contribution.
For given set of costs $\vc$ and some small $\varepsilon > 0$, consider the following mechanism:%\vspace{-0.5em}
    \begin{equation}\label{eqn:opt-mechanism}
        [\cM(\vm)]_i = 
            \begin{cases}
                a(m_i) \quad &\text{for } m_i \leq  m_i^\ast \\
                a(m_i^\ast) + (c_i + \varepsilon)(m_i - m_i^\ast)  \quad &\text{for } m_i \in [m_i^\ast, m_i^{\max}] \\
                a(\sum_j m_j)  \quad &\text{for } m_i \geq m_i^{\max} \,,
            \end{cases}
    \end{equation}
where $m_i^{\max}$ is defined such that  $a(m_i^{\max} + \sum_{j\neq i} m_j) = a(m_i^\ast) + (c_i + \varepsilon)(m_i^{\max} - m_i^\ast)$.    
We illustrate the mechanism in Figure~\ref{fig:accuracy-shaping-known-costs}.
Even without any external incentivization, agent $i$ will compute $m_i^\ast$ data points. Thus, for $m_i \leq m_i^\ast$ \eqref{eqn:data-max} returns a model trained on solely their own data. After $m_i^\ast$, however, the marginal gain in accuracy becomes smaller than the additional cost $c_i$. Hence, the agent requires active incentivization here and \eqref{eqn:data-max} ensures that for every additional data point computed, the marginal gain in accuracy is strictly more than the cost $c_i$. However, the mechanism cannot provide unlimited accuracy either and has to remain feasible, giving us our final constraint. 
\begin{restatable}[Data maximization with known costs]{theorem}{theoremoptmechanism}\label{thm:opt-mechanism}
    The mechanism $\cM$ defined by \eqref{eqn:opt-mechanism} has an unique Nash equilibrium and is data-maximizing for $\varepsilon \rightarrow 0^+$. At equilibrium, a rational agent $i$ will contribute $m_i^{\max}$ data points where $m_i^{\max} \geq m_i^\ast$, yielding a total of $\sum_j m_j^{\max}$ data points. 
    % Further, an agent's equilibrium data contribution is non-decreasing in their cost $c_i$.
\end{restatable}
% \todo{Proof.}

\begin{figure*}[!t]
\centering
\begin{tabular}{cccc} 
\includegraphics[width=0.3\textwidth]{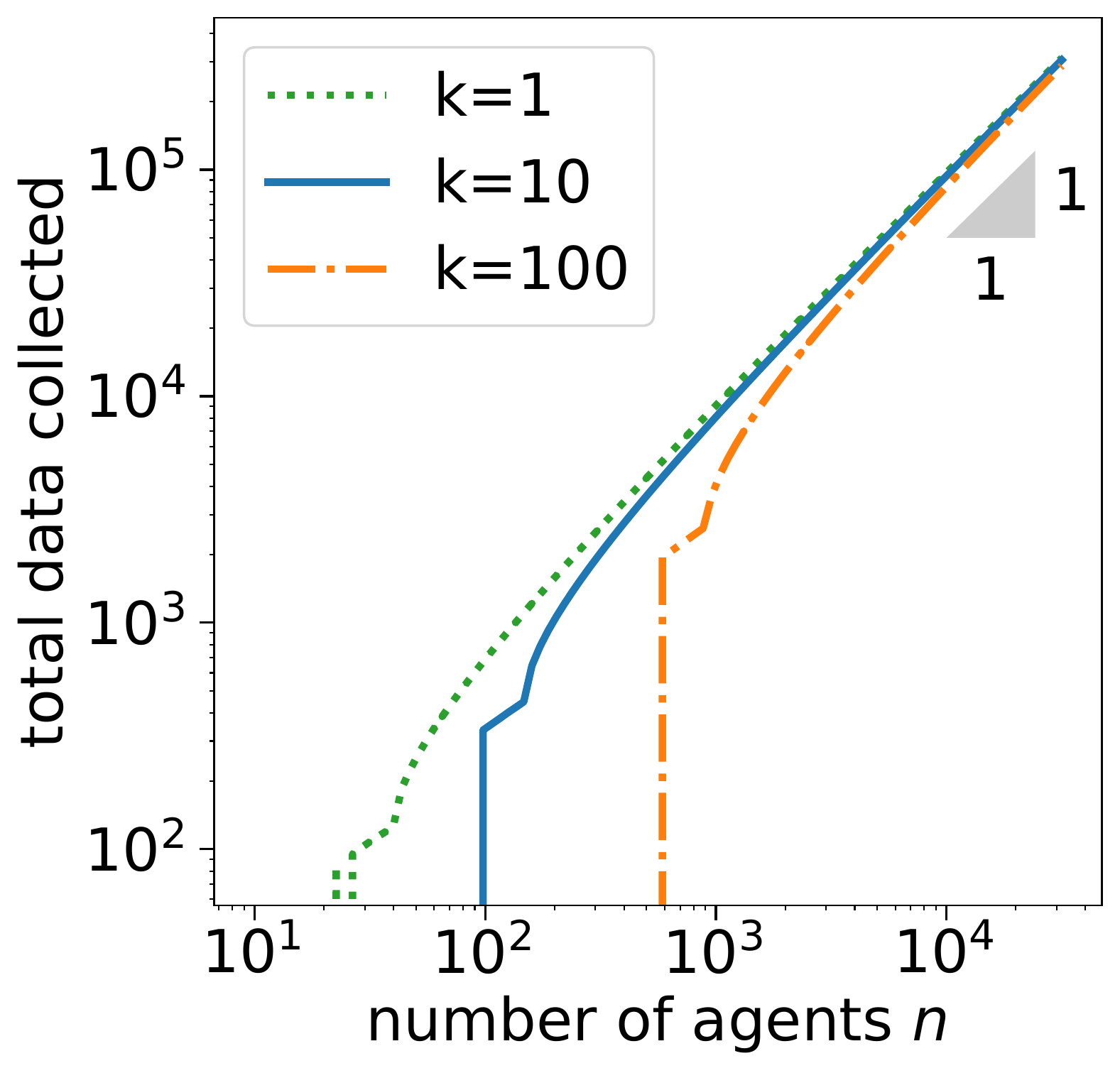} &
\includegraphics[width=0.3\textwidth]{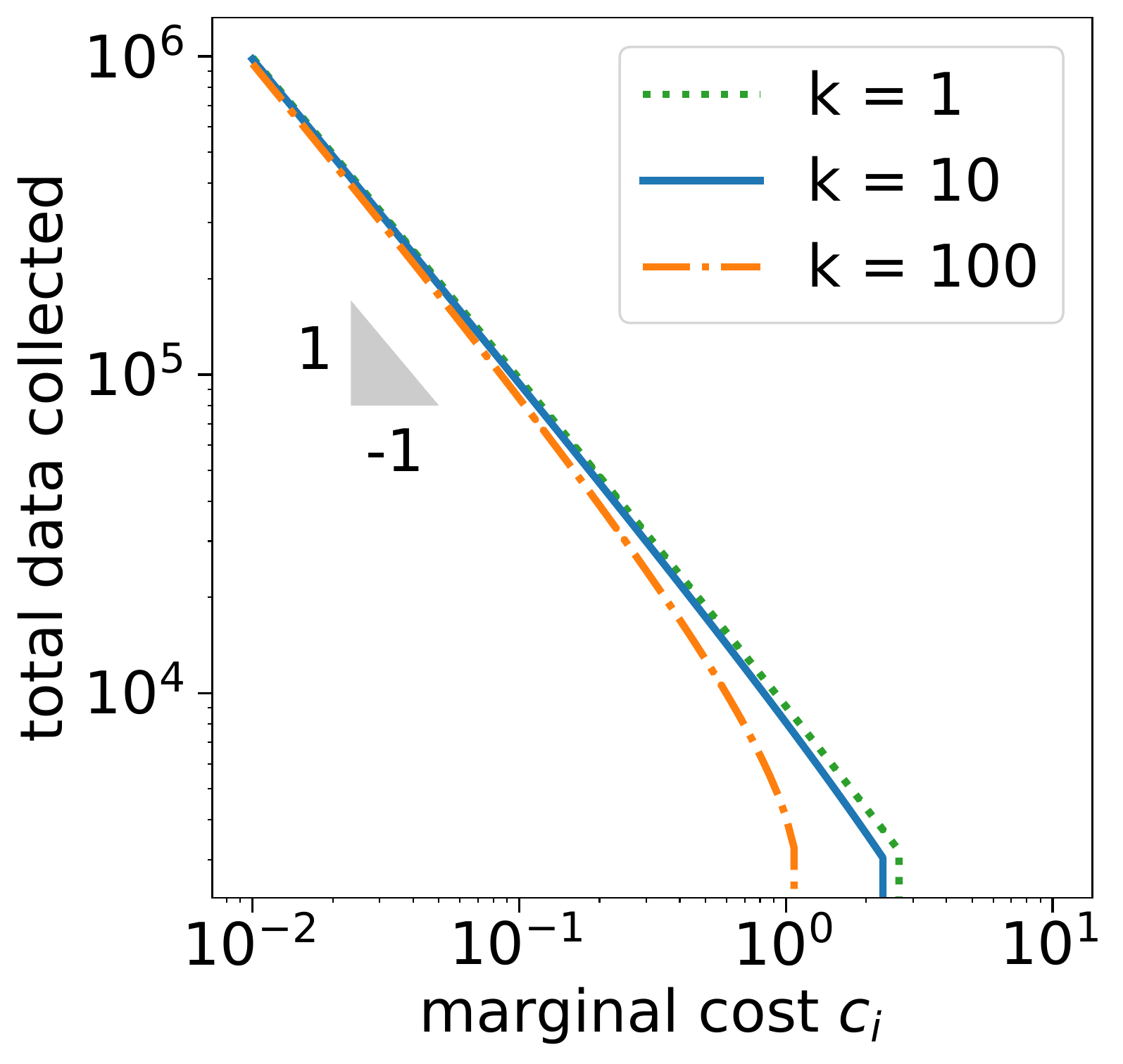} &
\includegraphics[width=0.3\textwidth]{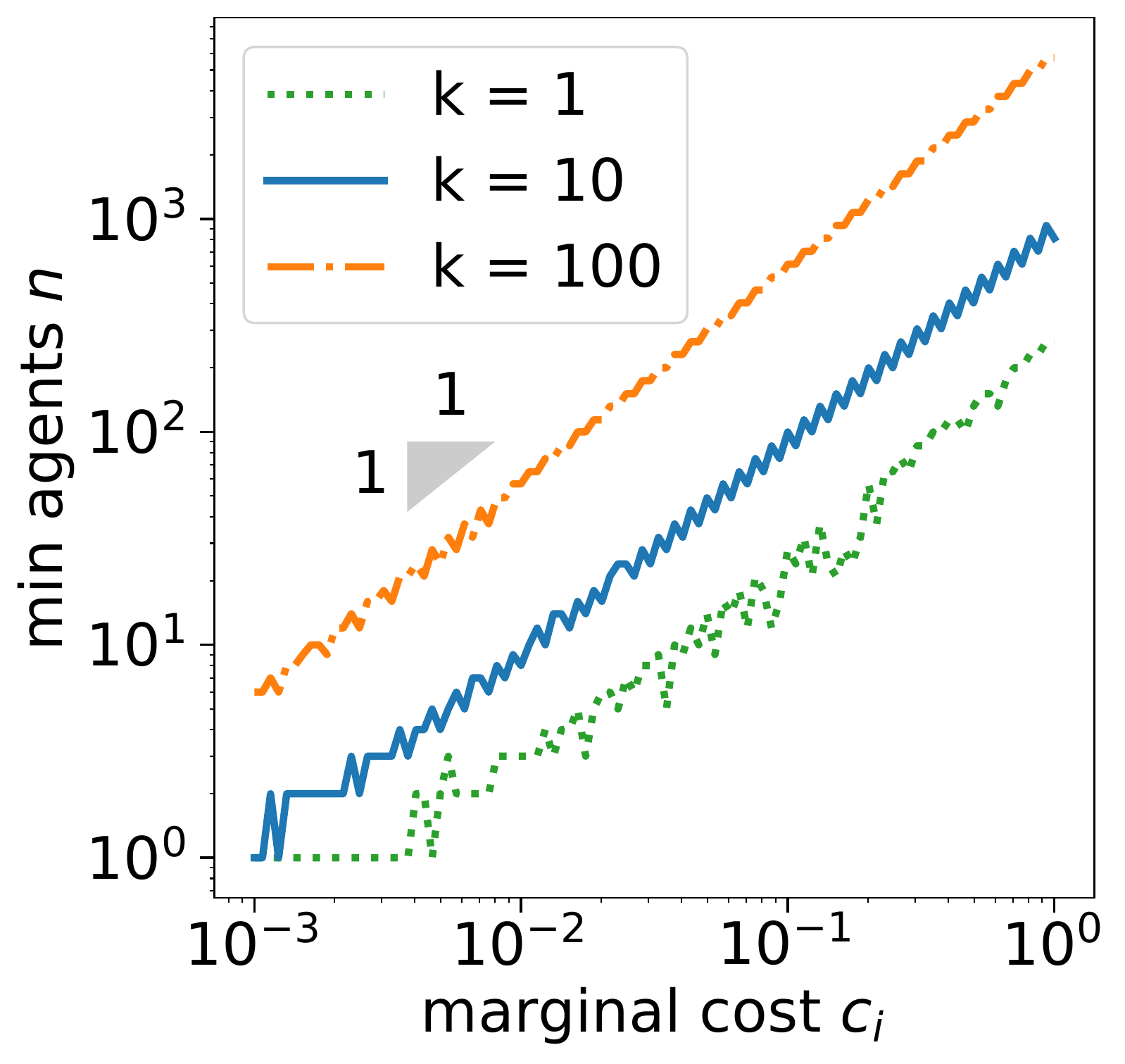}\\\vspace{-0.5em}
\scriptsize{$(a)$} & \scriptsize{$(b)$} & \scriptsize{$(c)$}
\end{tabular}%\vspace{-0.5em}
\caption{Equilibrium simulations of our mechanism: the total data collected at equilibrium  (a) increases linearly with the number of agents $n$, (b) decreases as $c_i^{-1}$ with increasing marginal cost per data point, and is relatively unaffected by the complexity $k$. (c) The number of agents required to cross the minimum viability threshold (i.e. smallest $n$ for which $m_i^{\max} > 0$) increases linearly with both the marginal cost $c_i$ and complexity $k$. Optimal individual contribution for all settings is $m_i^\ast = 0$; i.e., no data would be generated by standard federated learning.}
\label{fig:equil-cont} %\vspace{-1em}
\end{figure*}

For our example of autonomous-driving providers, the state agency can set up a central repository of the data to which the car providers are required to contribute. The cost for collecting a data point (say $c$) can easily be estimated and assumed to be the same for all agents. Using this estimate, the state agency can compute a threshold. The providers don't receive any additional data for contributing up to this threshold. For each data point contributed beyond the threshold, the providers receive an increasing amount of additional data. By Theorem~\ref{thm:opt-mechanism}, this would prevent free-riding by the providers and ensure the best trained model reaches the consumers.

\begin{remark}[Credible threat]
At equilibrium, mechanism $\cM$ in \eqref{eqn:opt-mechanism} ensures all agents contribute $m_i^{\max} \geq m_i^\ast$; \ie, they generate \emph{more} data than they would on their own. Further, every agent receives a model trained on this combined dataset with accuracy $a(\sum_j m_j^{\max})$. Thus, we can ensure that all agents fully utilize the combined data by using a threat that free-riding will be punished, even though at equilibrium such a threat is never actually invoked.
\end{remark}

% \begin{remark}[Scaling w.r.t.\ number of agents $n$]
Suppose all agents have the same cost $c$. Theorem~\ref{thm:opt-mechanism} shows that the mechanism collects $nm^{\max}$ data points in total. However, $m^{\max}$ also depends on $n$. This is because with a larger pool of data contributions, the server can more strongly incentivize an individual and extract more data. There is a natural ceiling to this though---the accuracy caps at $a_{opt} := \lim_{m \rightarrow \infty}a(m) \leq 1$. Thus, the absolute maximum data that can be extracted from an individual agent is $m$ which satisfies $a(m^\ast) + c(m - m^\ast) = a_{opt}$. This gives us the range for the total data contributions to be $\sbr*{nm^\ast ,\, n\rbr*{m^\ast + \nicefrac{(a_{opt} - a(m^\ast))}{c}}}$.
% \end{remark}

\begin{remark}[Overcoming minimum viability]\label{rem:min-viability} When $m^\ast = 0$, i.e., the problem is not solvable by an individual agent, the net contribution from our mechanism $nm^{\max}$ may still be positive. Suppose that the cost for all agents is the same $c$. Then, the total data collected is $m^{\text{tot}}$ which satisfies 
\[\nicefrac{c}{n} \cdot m^{\text{tot}}= a(m^{\text{tot}})\,.
\]
This implies that for sufficiently large $n$, the cost $c$ is successfully shared and we obtain a positive data contribution. However, note that $m^{\text{tot}} = 0$ is also a valid solution and remains an equilibrium. If all other agents don't contribute, there is no extra data to share and so there is no incentive to compute extra data. In practice, this undesirable equilibrium is unlikely to be encountered since it has lower utility. It can also be prevented by the platform itself taking part as an agent and committing to non-zero data collection.
\end{remark}

\begin{example}[Computing maximum contributions]\label{example:max-opt}
Consider the accuracy function arising from the generalization guarantees in Example~\ref{example:ERM}, and $n$ agents all with the same marginal cost $c$. For now, suppose that $m^\ast > 0$; i.e., the cost satisfies $c \leq a_{opt}^3/3k$. In this case, the data contribution from each agent becomes %%\vspace{-1em}
$ m^{\max} = \max\cbr[\big]{m\geq 0 \text{ s.t. } cm + 2 \sqrt{k / (nm)} \leq 3 (ck)^{1/3}}\,. 
$
For $n \geq 3$, the contribution can be lower bounded as $m^{\max} \geq \frac{3}{2}  \sqrt[3]{k/c^2}$. Further, when $n\rightarrow \infty$ we have $m^{\max} = 3\sqrt[3]{k/c^2}$. The total contribution is thus in the range $[\frac{3n}{2}  \sqrt[3]{k/c^2}\,,\, 3n \sqrt[3]{k/c^2}]$. In contrast, as we saw in Example~\ref{example:ind-opt}, an individual agent working alone would only collect $m^\ast =  \sqrt[3]{k/c^2}$ datapoints.

Next, consider the case where $m^\ast = 0$; i.e., the learning problem is too difficult for any individual. Using the computation in Remark~\ref{rem:min-viability}, the total contribution of the agents becomes at least $m^{tot} \geq \frac{a_{opt} n}{2c}$ as soon as we have a minimum number $n \geq 32ck / a_{opt}^3$ of participating agents. The total data collected increases with $n$ and decreases with $c$, but surprisingly is not affected by the complexity $k$. The complexity $k$ only imposes a constraint on the minimum number of agents required. Thus, in contrast to the failure of standard data-sharing, agents can successfully collaborate using our mechanism on otherwise insurmountable learning problems.

% \todo{We plot ... in Fig .. and make the following observations ...}
\end{example}
Empirically, in Figure~\ref{fig:equil-cont} we compute the equilibrium under the full utility function (see Eq~\eqref{eqn:generalization}) of our mechanism. We assume all agents have the same cost and see the effect of the equilibrium data contribution as we vary the cost $c$ and the total number of agents $n$. When not explicitly mentioned, we used the following default parameters: optimal accuracy of $a_{opt} = 0.95$, marginal cost $c_i = 0.1$, participants $n=10^4$. Under all parameter configurations of this experiment, the optimal individual contribution is $m_i^\ast = 0$, while the equilibrium data contributions are significantly larger as expected, validating our theory.

%\vspace{-0.5em}
\subsection{Incentive compatibility and distribution of surplus}%\vspace{-0.5em}
One of our motivating reasons for preventing free-riding was to ensure that none of the participating agents feel taken advantage of.  That is, we wish to satisfy some notion of fairness. However, there may potentially be new sources of unfairness in \eqref{eqn:opt-mechanism}. In particular, consider two agents, $i, j \in [n]$, with different costs:
% \begin{equation}
$    \text{if } c_i \leq c_j \text{ then } m_i^\ast \geq m_j^\ast\,.$
% \end{equation}
Here, an agent $i$ with smaller cost $c_i$ faces two disadvantages under mechanism \eqref{eqn:opt-mechanism}: (i) they have a larger threshold amount of data $m_i^\ast$ they have to contribute before receiving any benefit, and (ii) they receive a smaller increase in accuracy $(c_i + \varepsilon)$ for each additional data point computed. %Despite this, we show that our mechanism is not unfair.

If the cost for generating each data point is inherently fixed (such as the cost of driving a vehicle) this is arguably not an issue. However, in many other settings an agent may innovate and develop new methods to reduce their cost of collecting a data point. In fact, the  business model of large internet advertising providers is based on systems which can cheaply capture consumer data in order to show them better advertisements. Would our data-sharing mechanism \eqref{eqn:data-max} disincentivize agents from such innovations? We show this is in fact not true.
% \todo{move mechanism explanation.}
\begin{restatable}[Incentive compatibility]{theorem}{theoremunchangedutilities}
Under given costs $\vc$, consider our optimal mechanism \eqref{eqn:opt-mechanism} with equilibrium contributions $\vm^{\max}$, and agents working individually with equilibrium contributions of $\vm^\ast$. The utility of the every agent $i$ remains unchanged:
\[
    a(\textstyle\sum_{j}m^{\max}_j) - c_i m^{\max}_i = a(m_i^\ast) - c_i m_i^\ast\,.
\]
\end{restatable}
% \spk{We might want to make this a corollary since I think it directly follows from previous proof/construction.}
Thus, our mechanism does not induce any distortions in the incentive structure. Further, recall by Theorem~\ref{thm:ind}, the utility $u_i(m_i^\ast) \geq u_j(m_j^\ast)$ if $c_i \leq c_j$.  This implies that users with smaller costs continue to receive a higher utility, encouraging them to innovate and reduce the costs; i.e., our mechanism is incentive compatible. Of course this is assuming that the costs incurred by an agent is verifiable. They cannot lie about the true cost, but may be able to choose between different collection strategies.
\begin{remark}[Distribution of surplus]
One may ask where the additional surplus which is generated by agents collaborating has disappeared, since the agents receive none of it. Our mechanism utilizes this surplus in order to extract additional data, $m_i^{\max} - m_i^\ast$, from the agents. Thus, all the additional surplus goes into improving the accuracy of the model and hence to the end consumers of the model.
\end{remark}
Finally, in Appendix~\ref{sec:information_rent}, we also show how to extend our framework to the setting of unverifiable costs. Most of our conclusions translate to this setting as well.
% \paragraph{Data maximization.}

% In words, data maximization means that when all the agents are joining $M$, the total data they would be happy to collect is \textit{at least} the total amount of data that they would collect if they are training on their own. 

% \begin{definition}(data maximization)
% Denote the optimal amount of data that an agent $i \in [n]$ would collect when not joining the mechanism as $\cD_i^\ast$:
% \[
% \cD_i^\ast = \argmax_{\cD} u_i(\cD).
% \]

% Denote $\tilde \cD_i^\ast$ as the optimal amount of data for that agent when participating in the mechanism. Denote $\cD_{-i}$ as the data provided by all other agents $j \in [n], j \neq i$. Then: 
% \[
% \tilde \cD_i^\ast (M; \cD_{-i}) = \argmax_{\cD_i} u_i(M; \{\cD_i, \cD_{-i}\}).
% \]
% Then, a set of datasets $\textbf{D} = \{\cD_1, \cdots, \cD_n\}$ achieves \textit{data maximization} under mechanism $M$ if for $\forall i \in [n]$,
% \[
% \cD_ i = \tilde \cD_i^\ast (M; \cD_{-i}),
% \]
% and 
% \[
% |\cD_i^\ast| \leq |\cD_i|.
% \]
% \end{definition}

% \paragraph{Social welfare.}

% The social welfare $W(\{\cD_1, \cdots, \cD_n\})$ is:
% \[
% W(\{\cD_1, \cdots, \cD_n\}) = \sum_{i \in [n]} u_i(\cD_i).
% \]

% Conclusion %%%%%%%%%%%%%%%%%%%%%%%%%%%%%%%%%%%%%%%%%%%%%%%%%%%%%%%%%%%% 
\section{Discussion}\label{sec:discussion}
 We have initiated the study of mechanism design for data sharing, where the goal is to maximize the amount of data collected and the accuracy of the final model trained. We showed that the standard scheme of sharing each agent's data contribution with everyone else will inevitably lead to catastrophic free-riding where at most a single agent is contributing any data. In particular, this implies that when a learning problem is too difficult or expensive for a single agent to solve on their own, it will remain insurmountable under naive data sharing. Instead, more successful collaboration occurs if the data shared is tuned to the contributions of each of the agents. Our analysis suggests using a credible threat---agents receive the full benefit if they honestly submit data, but free-riding behavior is penalized by reducing the quality of the model returned. We presented examples of such schemes depending on whether the costs are verifiable or not. 
% Limitations
Our framework and results can also be extended to compensate for additional fixed joining costs, to allow the accuracy functions $\{a_i(m)\}$ to be different across the agents, and to employ general increasing convex cost functions $\{c_i(m)\}$. However, a significant limitation is that we needed the data points to be \emph{exchangeable}; i.e., for an agent $i$, every data point is identical in value and cost. A more general setting of heterogeneous data may exhibit qualitatively different behavior. We also assumed that the value function $a(m)$ is known and did not consider any privacy or data locality constraints. Under such constraints, we would have to also incorporate data verification into our mechanism design to prevent agents from submitting fake data points. 

Despite such limitations, our proposed framework already yields rich insights on the design of data-sharing mechanisms, and we believe that it has even broader implications. The federated learning community has recently been grappling with questions such as how to compensate the participants for their contributions~\citep{zhan2021survey}, incorporating incentives~\citep{zhan2021survey}, fairness constraints~\citep{shi2021survey}, and addressing privacy concerns~\citep{li2021survey}. While numerous mechanisms with varying definitions for each of these goals have been proposed, a principled manner to compare them is absent. We believe that the data-maximization viewpoint may provide such a principled framework---the idea is to come up with a mechanism which in the long-term induces agent behavior that maximizes data (or, more generally, value) creation.

% While we focused on free-riding in the context of data sharing, similar questions arise for more classical goods as well. Examples include content generation (books, music, videos), and innovation and research (drug discovery), etc. Here too, regulations such as copyrights or patents attempt preventing excessive free-riding in order to maximize good production. Our work might have broader applications in these settings as well.

% \begin{enumerate}
%     \item i.i.d. assumption - this work is more to illustrate importance of analyzing incentives. Real world is non i.i.d. and so probably has qualitatively different conclusions.
%     \item no privacy - data is fully shared to server.
%     \item joining costs.
%     % \item Examples of non-rivalorous + excludable: Data production in federated learning, innovation and research (drug discovery / ) , content generation (video / music / books).
%     \item Agents have independent identical utilities.
%     \item Fake data - verifiability.
% \end{enumerate}

% For final version
% \begin{ack}
%     Something else goes here
% \end{ack}

%%%%%%%%%%%%%%%%%%%%%%%%%%%%%%%%%%%%%%%%%%%%%%%%%%%%%%%%%%%%%%%%%%%%%

\bibliographystyle{plainnat}
\bibliography{bibliography}

\newpage
\appendix
% !TEX root = paper.tex

% CONTENTS

% {\hypersetup{linkcolor=black}
% \parskip=0em
% \renewcommand{\contentsname}{Contents of the Appendix}
% \tableofcontents
% \addtocontents{toc}{\protect\setcounter{tocdepth}{3}}
% }

\part*{Appendix}
% \section{Glossary}
% Surplus

% All notation

% Welfare

\section{Further Review on the Related Work and Contract Theory Background}\label{app:related_work}

The literature on mechanism design and federated learning is vast. We discussed the most closely related work in three verticals in the main text; we include a detailed review of the broader literature in this section.

Over the past decade, federated learning (FL) has emerged as an important paradigm in modern large-scale machine learning~\citep{konevcny2016federated, kim2019blockchained, kairouz2021advances, li2020review, rieke2020future}. Specifically, FL research has resulted in many applications to overcome practical challenges such as data silos and data sensitivity: on one side, since more training data often gives better model performance, data silos results in scarcity of labeled training data and puts limit on the industrial performance; on the other side, in high-stakes applications the data may contain private user information and thus the sharing of data is constrained by regulations and laws~\citep{voigt2017eu, baik2020data, cheng2020federated, mancini2021data}. Given these challenges, FL provides a useful scheme for different agents / parties to train collaboratively and leverage the benefit from other agents' data, while the training data remains distributed over the agents. Such a framework has been shown to be able to bring improved model performance to all the participants. Indeed, many prior works have been devoted to develop more scalable and communication-efficient distributed optimization algorithms for FL~\citep{konevcny2016federatedb,mcmahan2017communication, bonawitz2019towards}.

However, one cannot ignore an important aspect in the standard FL scheme, which is the incentives aspect. The standard FL scheme may incentivize strategic agents to contribute less data in order to minimize their data collection cost and maximize the gain from participating in the federated learning mechanism. Although the participation and contribution of each agent is often legislated by certain protocols, such free-riding behavior has been notoriously hard to regulate and prevent in practice~\citep{fraboni2021free, huang2020exploratory}. Recently, a few works have started to explore such free-riding behavior in FL, with various incentive models proposed~\citep{richardson2020budget, sarikaya2019motivating,lin2019free,fraboni2021free,ding2020incentive, zhang2022enabling}. However, the majority this work has focused on a taxonomy of free-rider attacks or the detection of attacks under the existing FL scheme, instead of proposing mechanisms that incentivize maximal data contribution. In this work, we strive for a mechanism for information sharing under the standard federated learning setting such that rational agents are incentivized to contribute their maximal amount of data.

In this work, we focus on the free-riding behavior of FL agents in terms of data collection. In FL, the data collection happen on the agents' side before they join the mechanism for training models. Therefore, the cost of collecting data is often \textit{private} information to each agent. Such an information asymmetry brings difficulty to prevent free-riding, because the agents might simply report fake costs. This brings the need to design \textit{incentive} mechanisms for FL, under which the agents are incentivized to behave truthfully, which is also guaranteed to lead to the best utility. 

Indeed, designing incentive mechanisms under private costs is not new, and has been a main focus of the contract theory literature~\citet{smith2004contract, laffont2009theory, bolton2004contract}. Moreover, the existence of a central server (a ``principal") in FL brings further convenience to apply a principle-agent model. An emerging line of recent works have been exploring the application of contract theory for federated learning~\citep{kang2019toward,kang2019incentivea,kang2019incentiveb, lim2021towards, tian2021contract, cong2020game, zhan2020learning}. In particular,~\citet{tian2021contract} proposed a contract-based aggregator under a multi-dimensional contract model over two possible types of agents and showed improved model generalization accuracy under that contract. However, their mechanism focused on eliciting the private type information instead of maximizing the data contribution. To the best of our knowledge, our work is a first step to use contract theory for \textit{data maximization} in federated learning. Further, prior work has focused on how to design payments to agents, rather than the accuracy-shaping problem that we focus on here.

This work is related to the active line of research on mechanism design for collaborative machine learning, which involves multiple parties each with their own data, jointly training a model or making predictions in a common learning task~\citep{sim2020collaborative, xu2021gradient}. 
In collaborative machine learning, a major focus has been the design of model rewards (i.e., data valuation) in order to ensure certain fairness or accuracy objectives. Towards that goal, there has been model rewards proposed based on notions from the cooperative game theory literature such as the Shapley value~\citep{jia2019towards, wang2020principled}. However, the guarantees of these model rewards depend on the assumption that the agents are already willing to contribute the data they have. In this work, we study a different incentivization task for data maximization.

More broadly, apart from data maximization, there are  other objectives which are of interest in federated learning, such as fairness and welfare objectives, that have been under active study~\citep{donahue2021optimality, donahue2020model, mohri2019agnostic}. We defer a thorough analysis of the tradeoffs among various objectives to future research.

% Unknown costs %%%%%%%%%%%%%%%%%%%%%%%%%%%%%%%%%%%%%%%%%%%%%%%%%%%%%%%%%%%% 

\section{Data Maximization under Unverifiable Costs}\label{sec:information_rent}
Until now, we assumed that the cost of all agents is known to everyone involved, or is atleast verifiable. In some settings where the costs are universal and outside the control of the agent, this assumption may be justified. 
% \begin{wrapfigure}{r}{8cm}
%     \centering
%     %\vspace{-1.5em}
%     \includegraphics[width=\linewidth]{figures/figure4.pdf}%\vspace{-3.5em}
%     \caption{Accuracy shaping mechanism under unknown costs. \textit{(red curve)}: model accuracy returned to agent $i$ by the mechanism; \textit{(grey curve)}: model accuracy for agent $i$ without participation; \textit{(green curve)}: model accuracy if agent $i$ receives all the data from the other agents.}%\vspace{-2em}
%     \label{fig:accuracy-shaping-unknown-costs}
% \end{wrapfigure}
However, in numerous other cases, the exact process of the data generation may be a trade secret and so there is uncertainty about the cost incurred by an agent. In this section, we examine how to incorporate such uncertainties into our mechanism.

% In such cases, we can only make a guess about the true cost and we will necessarily have some uncertainty. \wg{feel that we might not want to say "make a guess"}

We focus on the simplest version of this uncertainty. Suppose that we know that the cost of each agent can either be low ($\low c$) or high ($\high c$). Further, suppose we have some prior knowledge where agent $i$ has low cost $\low c$ with probability $p_i$ and $\high c$ with $(1-p_i)$. Note that there is an inherent \emph{information asymmetry} in this setting. The agent knows the realization of their cost, $c_i \in \cbr{\high c, \low c}$, whereas the server only knows the distribution from which it was drawn. In particular, the server needs to present a mechanism $\cM$ to an agent without knowing their actual cost.
%\vspace{-1em}
% We show how to collect the maximum expected data contribution, given that the costs of each agent are sampled from our prior distribution.

% \todo{move mechanism explanation.}
\subsection{Mechanism description}
Suppose each agent independently selects their cost to be low $(c_i = \low c)$ with probability $p_i$. Let an agent with low cost $\low c$ generate $\low m^\ast$ data points at equilibrium on their own (and correspondingly define $\high m^\ast$ for a high-cost agent). Then, for some small $\varepsilon > 0$ and $\high m^\ast \leq m_i^\uparrow \leq m_i^\downarrow$, consider the following mechanism (illustrated in Figure~\ref{fig:accuracy-shaping-unknown-costs})%\vspace{-1.1em}
    \begin{equation}\label{eqn:unknonw-opt-mechanism}
        [\cM(\vm)]_i = 
            \begin{cases}
                a(m_i) \quad &\text{for } m_i \leq  \high m^\ast \\
                a(\high m^\ast) + (\high c + \varepsilon)(m_i - \high m^\ast)  \quad &\text{for } m_i \in [\high m^\ast, m_i^\uparrow] \\
                a(m_i^{\downarrow} + \sum_{j\neq i} m_j)  - (\low c + \varepsilon)(m_i^\downarrow - m_i)  \quad &\text{for } m_i \in [m_i^\uparrow, m_i^\downarrow] \\
                a(\sum_j m_j)  \quad &\text{for } m_i \geq m_i^\downarrow \,.
            \end{cases}
    \end{equation}
    \begin{wrapfigure}{r}{8cm}
    \centering
    \vspace{-1.1em}
    \includegraphics[width=\linewidth]{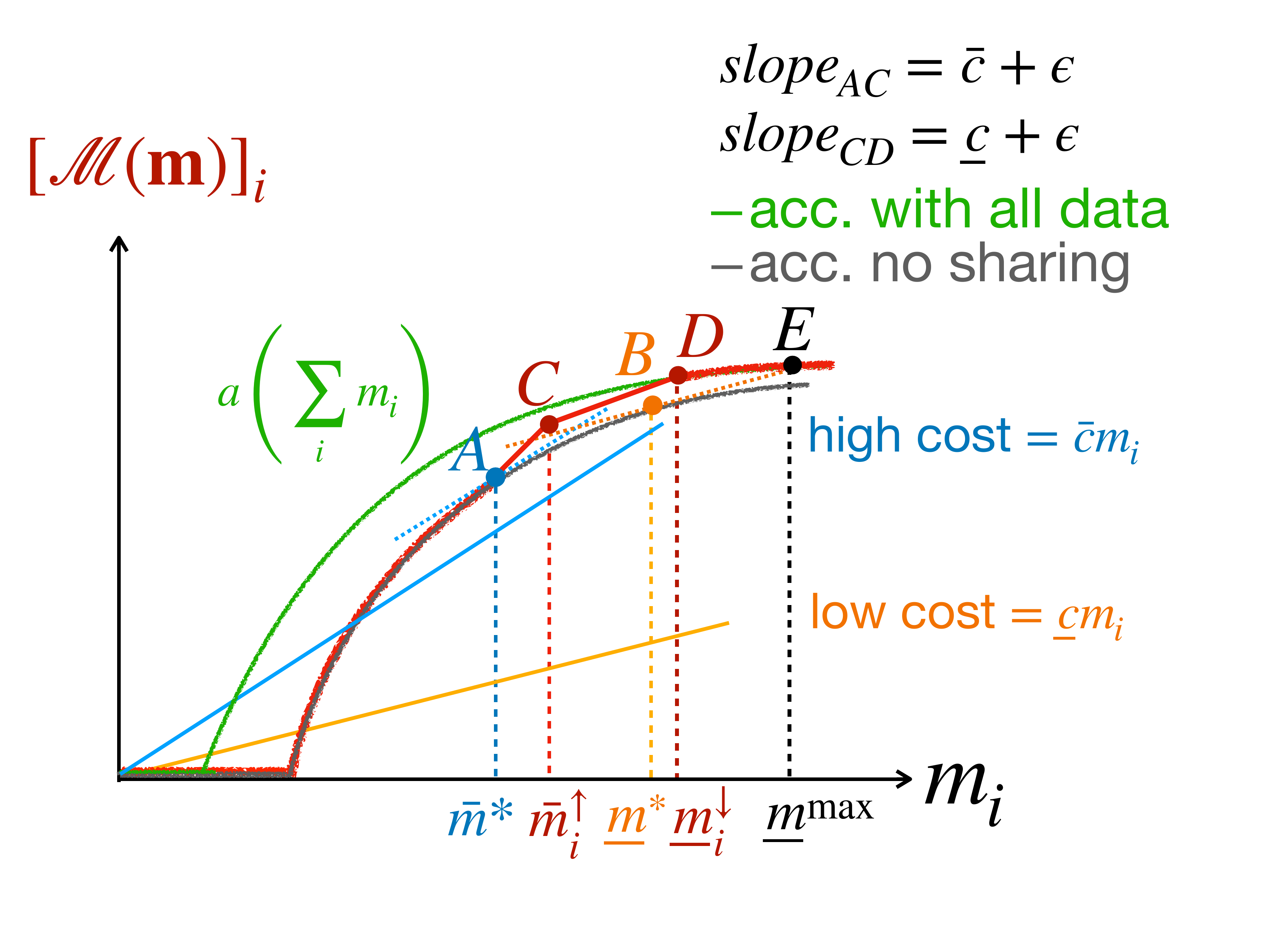}\vspace{-1.1em}
    \caption{Accuracy shaping mechanism under unknown costs. \textit{(red curve)}: model accuracy returned to agent $i$ by the mechanism; \textit{(grey curve)}: model accuracy for agent $i$ without participation; \textit{(green curve)}: model accuracy if agent $i$ receives all the data from the other agents.}%\vspace{-2em}
    \label{fig:accuracy-shaping-unknown-costs}
\end{wrapfigure}
    Recall from Theorem~\ref{thm:ind} that $\low m^\ast \geq \high m^\ast$ since $\low c \leq \high c$. Thus, agents with either costs do not need additional incentive to collect data up to $\high m^\ast$. Now, consider a high-cost agent. After $\high m^\ast$, they need a marginal gain in accuracy of at least $\high c$ which they do not get on their own. Additional supplementary data is provided by \eqref{eqn:unknonw-opt-mechanism} until $m_i^\uparrow$ to incentivize a high-cost agent. It is now in their best interest to contribute $m_i^\uparrow$. For the low-cost agent, the marginal gain in accuracy is at least $\low c$ until $m_i^\downarrow$, making this their best contribution. The specific values of $m_i^{\downarrow}$ and $m_i^{\downarrow}$ (points D and C) can then be chosen to maximize the expected data contribution $((1-p_i) m_i^\uparrow + p_i m_i^\downarrow)$.
    
    For $\Delta m_{\text{-}i} := \sum_{j \neq i}m_j$, let $\high m^{\max}$ be the maximum amount of data a high-cost agent can be incentivized to contribute as in \eqref{eqn:opt-mechanism} \ie it is defined to be $a(\high m^{\max} + \Delta m_{\text{-}i}) = a(\high m^\ast) + \high c(\high m^{\max} - \high m^\ast)$, and $\low m^{\max}$ defined correspondingly for the low-cost agent. Then,  we define $m_i^\downarrow$ (point D) to satisfy%\footnote{$\clip(a, [b,c]) = \min(\max(a, b), c)$}
    \begin{equation}\label{eqn:unknown-compromise}
        a'(m_i^\downarrow + \Delta m_{\text{-}i}) = \min\rbr[\Big]{\max\rbr[\Big]{ \low c - \tfrac{p }{1-p}\high c \ ,\  a'(\low m^{\max} +\Delta m_{\text{-}i}} \,,\, a'(\high m^{\max} +\Delta m_{\text{-}i})}\ .
    \end{equation} 
    Then, we can define $m_i^\uparrow$ (point C) as the intersection of the two linear curves (starting from A and D in Fig~\ref{fig:accuracy-shaping-unknown-costs}): %\vspace{-1em}
    \begin{equation}\label{eqn:unknown-feasibility}
    a(m_i^{\downarrow} + \textstyle\sum_{j\neq i} m_j)  - (\low c + \varepsilon)(m_i^\downarrow - m_i^\uparrow) = a(\high m^\ast) + (\high c + \varepsilon)(m_i^\uparrow - \high m^\ast) \,.
    \end{equation}
    Note that our mechanism withholds some data from a high-cost agent resulting in a lower accuracy model for them. This is necessary to prevent a contribution level targeted at high-cost agent from becoming attractive to a low-cost agent. 
    
    \subsection{Analysis} We now analyze the properties of our expected data-maximization algorithm.
\begin{restatable}[Expected data maximization]{theorem}{theoremunknownoptmechanism}\label{thm:unknonw-opt-mechanism}
    Mechanism \eqref{eqn:unknonw-opt-mechanism}  is feasible, satisfies IR, and has a unique Nash equilibrium: $m_i^\text{eq} = m_i^\uparrow$ if $c_i = \high c$ and otherwise $m_i^\text{eq} = m_i^\downarrow$.
    % %\vspace{-0.5em}
    % \begin{equation}
    %     m_i^\text{eq} = \begin{cases}
    %         m_i^\uparrow &\text{ if agent $i$ has high cost } c_i = \high c\\
    %         m_i^\downarrow &\text{ if agent $i$ has low cost } c_i = \low c\,.
    %     \end{cases} 
    % \end{equation}
    Further, for $\varepsilon \rightarrow 0^+$, the mechanism \eqref{eqn:unknonw-opt-mechanism} maximizes the expected (over the sampling of the true costs) amount of data collected with%\vspace{-0.5em}
    \[
        \textstyle\sum_j (1-p_j) m_j^\uparrow + p_j m_j^\downarrow = \max_{\cM}\cbr*{ \textstyle\sum_j \E_{\vc}\sbr{m_j^\cM}\,, \text{ subject to $\cM$ being feasible and IR}}\,.
    \]
\end{restatable}
% \todo{Add figure with accuracy curve and connect to explanation below.}
% We illustrate this mechanism in Figure~\ref{fig:accuracy-shaping-unknown-costs}.

\begin{remark}[Decreased data collection]
By construction of our mechanism, the contribution of a high-cost agent would be $m_i^\uparrow \in [\high m^\ast, \high m^{\max}]$ \ie they contribute more than they would on their own, but lesser than the max possible under known costs. Further, our assumption that $a(\cdot)$ is concave means $a'(\cdot)$ is non-increasing. Hence, \eqref{eqn:unknown-compromise} implies that the data contributed by a low-cost agent is $m_i^\downarrow \in [\high m^{\max}, \low m^{\max}]$. However, if $p_i \geq \frac{\low c}{\low c + \high c}$, \eqref{eqn:unknonw-opt-mechanism} always implies that $m_i^\downarrow = \low m^{\max}$.% and they contribute as much as before.
% This would happen if the likelihood $p_i$ the overall expected data collected is still more than what the agents would have contributed on their own.
\end{remark}
% \begin{remark}[high-cost agent]
We extract lesser data than if we knew the agent's true cost \ie $m_i^\uparrow \leq \high m^{\max}$. However, they also receive a model which has worse accuracy with $a(\high m^\ast) + \high c(m_i^\uparrow - \high m^\ast) \leq a(\sum_j m_j)$ \ie it is not trained on the combined data. This is because if we offered a full accuracy model to a high-cost agent at $\high m^{\max} \leq m_i^\downarrow$ contribution, the low-cost agent can claim they are actually high-cost and cheat our system.
% \end{remark}
Instead, now the low-cost agent will contribute $m_i^\downarrow \geq \high m^{\max}$ and will receive a model trained on the combined data with accuracy $a(\sum_{j} m_j)$. %Here as well, the amount of data collected is smaller than what it could have been if we knew the agent's true cost with $m_i^\downarrow \leq \low m^{\max}$.
\begin{restatable}[Information rent]{theorem}{theoreminformationrent}
Consider our optimal mechanism \eqref{eqn:unknonw-opt-mechanism} with equilibrium contributions $m_i^{eq} = m_i^\uparrow$ for a high-cost agent and $m_i^{eq} = m_i^\downarrow$ for the low-cost agent. Further, let $\high m^\ast$ and $\low m^\ast$ be the equilibrium individual contributions. Then, the utility of the high-cost agent remains unchanged with %%\vspace{-1em}
$
    a(\high m^\ast) + \high c(m_i^\uparrow - \high m^\ast) - \high c m_i^\uparrow = a(\high m^\ast) - \high c \high m^\ast\,.
$
The utility of a low-cost agent, however, improves by $\left(\low c(\low m^{\max} - m_i^\downarrow) - a(\low m^{\max} + \Delta m_{\text{-}i}) + a( m_i^\downarrow+ \Delta m_{\text{-}i}) \right) \geq 0$.
\end{restatable}
% \begin{remark}[Distribution of surplus]
Because a low-cost agent can always lie and pretend to be high cost, they hold some power over the server when $m_i^\downarrow < \low m^{\max}$. This is reflected in the extra utility they manage to extract and is called information rent. The utility of the high-cost agent remains unchanged since they hold no such power.
% \end{remark}

\begin{example}[Computing decreased contributions]\label{example:unknown-max-opt}
Consider the accuracy function arising from the generalization guarantees in Example~\ref{example:ERM}, and $n$ agents whose marginal cost is chosen uniformly (i.e. $p_i = 1/2$) from the set $\{\low c, \high c \}$. Also suppose that $m^\ast > 0$ i.e. the cost satisfies $c \leq a_{opt}^3/3\low c$. In this case, we know from Example~\ref{example:ERM} that the individual contribution is $\high m^\ast = \sqrt[3]{k/(\high c)^2}$ and correspondingly $\low m^\ast = \sqrt[3]{k/\low c^2}$. For $p_i = 1/2$, \eqref{eqn:unknonw-opt-mechanism} implies that $m^\downarrow = \low m^{\max}$. Further, the contribution of the high-cost agent is $m^\uparrow = \gamma_t \high m^{\ast} + (1 - \gamma)\low m^{\ast}$ for some $\gamma \in [0,1]$. Thus, the high-cost agent's contribution is in between the individual contributions, while the low-cost agent contributes the maximum amount they would have even if we knew their true cost. 
% \todo{clean up this example.}
% \todo{We plot ... in Fig .. and make the following observations ...}
\end{example}

\section{Proofs from Section~\ref{sec:framework} (Optimal Individual Contributions)}

% \wg{@SPK: Theorem 1 and the proof need a double check. }
\individualtheorem*
\begin{proof}
Recall that the utility function (see Eq.~\ref{eqn:utility}) of a single agent is:
\[
u_i(m_i) = a(m_i) - c_i m_i. 
\]
Thus we have,
\[
u_i'(m_i) = a'(m_i) - c_i. 
\]
Denote $\argmax_{m} a(m) =0$ as $m^0$. By definition, $\forall m_i > m^0$, $a(m_i) = b(m_i) > 0$.
Given that $b(\cdot)$ is concave, $b'(m_i), m_i \geq m^0$ (or $u_i'(m_i)$) is maximized when $m_i = m^0$. 
\paragraph{Case 1 (high-cost agent): $u_i'(m^0) \leq 0$.} Then, for $\forall m_i \geq m^0$, $u_i'(m_i) \leq u_i'(m^0) \leq 0$. On the other hand, $\forall 0 \leq m_i \leq m^0$, $u_i'(m_i) = -c_i \leq 0$. Thus $u_i(m_i)$ is non-increasing, and $m_i^\ast = 0$. The utility function of an agent in this case is illustrated in Figure~\ref{fig:single-agent} (d).
\paragraph{Case 2 (mid-cost agent): $u_i'(m^0) > 0$ and $\max_{m_i} u_i(m_i) \leq 0$.} When $u_i'(m^0) > 0$, that implies that at $m^0$, $b'(m^0) > c_i$. Moreover, for $m_i \geq m^0$, we have that 
\[
u_i'(m_i) = b'(m_i) - c_i < b'(m^0) - c_i.
\]
Therefore, since $b(m_i)$ is concave, it is possible that $u_i(m_i)$ increases first after $m^0$. However, as long as $\max_{m_i} u_i(m_i) \leq 0$, we still have that $m_i^\ast = 0$. The utility function of an agent in this case is illustrated in Figure~\ref{fig:single-agent} (c).

\paragraph{Case 3 (low-cost agent): $u_i'(m^0) > 0$ and $\max_{m_i} u_i(m_i) > 0$.} Recall that for a mid-cost agent, it is possible that $u_i(m_i)$ increases first after $m^0$. Moreover, given that $a(m_i) \leq 1$, as $m_i \rightarrow \infty$, 
\[
u_i'(m_i) = b'(m_i) - c_i \leq 0.
\]
Therefore, there exists $\alpha_i^\ast > m^0 > 0$ such that $b'(\alpha_i^\ast) = c_i$. The utility function of an agent in this case is illustrated in Figure~\ref{fig:single-agent} (b).

Combining the three cases above completes the first part of the proof. 

Next, consider two agents with costs $c_i \leq c_j$. Note that for any fixed $m$, $u_i(m) \geq u_j(m)$. Hence, the inequality also holds after minimizing both sides. Finally, note that if $j$ is not a low-cost agent, it is clear that $m_i^\ast \geq m_j^\ast = 0$. If both $i$ and $j$ are low-cost agents, note that $m_i^\ast = b'^{-1}(c_i)$ and $m_j^\ast = b'^{-1}(c_j)$. Since $b(\cdot)$ is concave and positive, $b'$ (and hence $b'^{-1}$) is non-increasing. This implies that $m_i^\ast \geq m_j^\ast$ finishing the theorem.
\end{proof}

\section{Proofs from Section~\ref{sec:free_ride} (Modeling Multiple Agents and Catastrophic Free-riding)}
\equlirbiumtheorem*
\begin{proof}
For a set of contributions $\vm$, define the following best response mapping:
\begin{equation}\label{eqn:argmax-equil}
    [B(\vm)]_i := \argmax_{\tilde m_i\geq 0}  \cbr*{u_i(\tilde m_i, \vm_{-i}) := [\cM( \tilde m_i, \vm_{-i})]_i - c_i \tilde m_i }\,,
\end{equation}
where recall $[\cM( \tilde m_i, \vm_{-i})]_i$ is the accuracy returned by the mechanism upon agent $i$ submitting $\tilde m_i$ data points and the rest contributing $\vm_{-i}$. Note that the mapping defined above is a multi-valued function i.e. $B: \R^n \rightarrow 2^{\R^n}$. This is because the $\argmax$ defined above may potentially return multiple values. Nevertheless, suppose that there existed a fixed point to the mapping $B$ i.e. there existed $\tilde \vm$ such that $\tilde \vm \in B(\tilde \vm)$. Then, $\tilde \vm$ is the required equilibrium contribution since by definition of the arg-max  we have for any $m_i \geq 0$, 
\[
    [\cM( \tilde m_i, \tilde\vm_{-i})]_i - c_i \tilde m_i \geq [\cM( m_i, \tilde\vm_{-i})]_i - c_i m_i\,.
\]
So, we only have to prove that the mapping $B$ has a fixed point. Since the mechanism $\cM$ is feasible, by Definition \ref{def:feasible} and equation \eqref{eqn:acc} we have
\[
        [\cM( \vm)]_i \leq a(\textstyle\sum_j m_j) \leq \lim_{m \rightarrow \infty} a(m) \leq 1\,.
\]
This implies that
\[
    0 \geq u_i(\tilde\vm) \leq 1 - c_i \tilde m_i \Rightarrow \tilde m_i \leq \nicefrac{1}{c_i}\,.
\]
Thus, we can restrict our search space to a \emph{convex} and \emph{compact} product set $\cC := \prod_{j}\sbr*{0, \nicefrac{1}{c_i}} \subset \R^n$ and our mapping is then over $B: \cC \rightarrow 2^{\cC}$. Next by assumption on the mechanism $\cM$, our utility function can be written as
\[
    u_i(m_i, \vm_{-i}) = \max(-c_i m_i\,,\, \nu(m_i, \vm_{-i}) -c_i m_i )\,,
\]
where $\nu(m_i, \vm_{-i}) -c_i m_i$ is concave in $m_i$. Unfortunately, $u_i$ may not be quasi-concave in $m_i$ because of the max. If it was quasi-concave, the mapping $B(\vm)$ would be continuous in $\vm$ and applying Kakutani’s theorem would yield the existence of the required fixed point (see~\citet{maskin1986existence} or \citet[Lecture 11]{lecture11}, for details). 
\begin{lemma}[Kakutani’s fixed point theorem]\label{lem:fixed}
Consider a multi-valued function $F: \cC \rightarrow 2^{\cC}$ over convex and compact domain $\cC$ for which the output set $F(\vm)$ i) is convex and closed for any fixed $\vm$, and ii) changes continuously as we change $\vm$. For any such $F$, there exists a fixed point $\vm$ such that $\vm \in F(\vm)$.
\end{lemma}

However, our utility function is not quasi-concave and the mapping $B$ may be discontinuous. While there have been recent extensions of Kakutani’s fixed point theorem to half-continuous functions (e.g. \citet[Theorem 3.2]{bich2006some}), the mapping $B$ does not satisfy this either. We next study the exact nature of discontinuity.

\begin{lemma}\label{lem:discont-mapping}
Consider the best-response mapping $B$ in \eqref{eqn:argmax-equil} over convex and compact domain $\cC$. For any $\vm$, either the mapping $[B(\vm)]_i$ is convex, closed, and continuous in $\vm$, or $0 \in [B(\vm)]_i$. 
\end{lemma}

\begin{figure}
    \centering
    \includegraphics[width=0.65\linewidth]{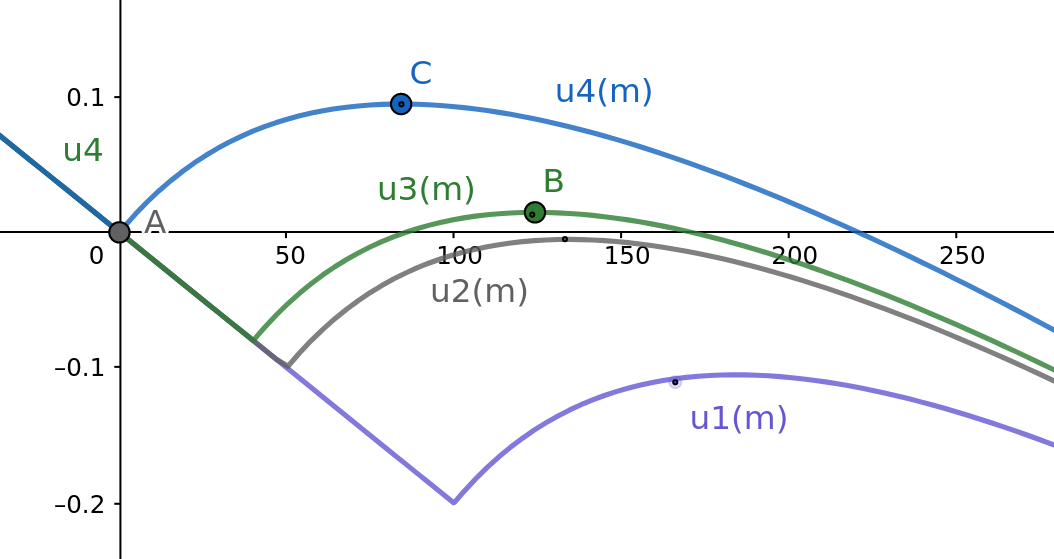}
    \caption{Utility curves $u(m_i)$ of some agent $i$, and the corresponding discontinuous best responses ($\tilde m_i \geq 0$ which maximizes $u(m_i)$). For both $u1$ and $u2$, the best response of the agent is $\tilde m_i = 0$ (point A), and points B and C are the best responses for $u3$ and $u4$. A small change in the utility curves (from $u2$ to $u3$) can result in a large change in the best response (from A to B).}
    \label{fig:continuity}
\end{figure}

\begin{proof}
Figure~\ref{fig:continuity} looks at the best response mapping $B_i$ depending on the utility curve $u_i(\cdot, \vm_{-i})$. Even if the utility itself is smoothly varying with the parameters $\vm$, the best response may be discontinuous. In Fig.~\ref{fig:continuity}, for a small change in the utility curve between $u2(m_i)$ to $u3(m_i)$, the best response drastically changes from $\tilde m_i =0$ (A) to $\tilde m_i > 125$ (B). However, this is the only source of discontinuity.

Recall that our utility function $u_i$ is a max of a decreasing linear function and a concave function. Thus it has at most two local maxima: either 0, or the maxima of the concave function $f(m_i) = \nu_i(m_i; \vm_{-i}) - c_i m_i$. The set of maxima of a continuous concave function is continuous, closed and convex. Hence, either $0$ is part of the best response, or $[B(\vm)]_i$ is continuous, closed and convex.
\end{proof}

Armed with Kakutani's fixed point theorem Lemma~\ref{lem:fixed} and a description of the discontinuities in the best response mapping Lemma~\ref{lem:discont-mapping}, we can continue with the proof of existence of a fixed point for $B$. 
Given any index set $\cI \subseteq [n]$, we can define the following sub-domain $\cC_{\cI} := \prod_{i \in \cI} [0, \nicefrac{1}{c_i}] $. Given any vector $\vp \in \cC_{\cI}$, we can construct its extension $m(\vp; \cI) \in \cC$ as 
\[
    [m(\vp; \cI)]_i := \vp_i \text{ if } i \in \cI \text{, and } 0 \text{ otherwise.}
\]
We will omit the $\cI$ dependence and use $m(\vp)$ when clear from context. Given this mapping between sub-domain $\cC_{\cI}$ and the full domain $\cC$, we can define a mapping:
\[
    B_{\cI}(\vp): \cC_{\cI} \rightarrow 2^{\cC_{\cI}} := \rbr*{[B(m(\vp))]_i \text{ for } i \in \cI}
\]
Finally, for any $\vm \in \cC$, define the set of indices $\cI(\vm) \subseteq [n]$ as
\[
    \cI(\vm) := \{i \text{ for which } 0 \notin [B(\vm)]_i \,.\}
\]
Let us start from $\vm = \mathbf{0}$. If $\cI(\mathbf{0}) = \emptyset$, we are done since this implies $\mathbf{0} \in B(\textbf{0})$. Otherwise, Lemma~\ref{lem:discont-mapping} states that the mapping $B_{\cI(\mathbf{0})}(\vp)$ over the compact convex domain $\cC_{\cI(\mathbf{0})}$ is convex, compact and continuous. Hence, by Lemma~\ref{lem:fixed}, it has a fixed point such that $ \vp^1 \in B_{\cI}(\vp^1)$. We can inductively continue applying the same argument. If $m(\vp^1)$ is a fixed point of the full mapping $B$ with $m(\vp^1) \in B(m(\vp^1))$, we are done. Otherwise, $\cI(m(\vp^1)) \supset \cI(\mathbf{0})$ and we can continue repeating the same argument inductively. Since the size of $\cI$ is at most $n$, the recursion will stop and yield a fixed point $\tilde \vm \in \cC$ such that $\tilde\vm \in B(\tilde\vm)$. As we initially proved, this fixed point $\tilde \vm$ to the best response dynamics is also the equilibrium of our mechanism.

% In other words, as long as there is always a hyperplane $\vp$ which robustly separates $\vm$ and its map $F(\vm)$ i.e. it separates all small perturbations $\vm'$ and maps $B(\vm')$, the mapping $F$ has a fixed point.

% Since the mapping $B$ is half-continuous, proves that it has a fixed point which (as we previously showed) corresponds to a pure equilibrium.
\end{proof}

\federatedtheorem*
\begin{proof}
Let $\tilde i$ be the agent with the least cost. By Theorem~\ref{thm:ind}, it follows that $m_j^\ast \leq m_{\tilde i}^\ast$ for all agents $j \in [n]$. 

First, suppose that $m_{\tilde i}^\ast > 0$. In this setting, all other agents agent $j$ will have access to data contributed by $\tilde i$ which is $m_{\tilde i}^\ast$. Now given access to this, the marginal gain in accuracy for an additional data-point for any agent $j$ is less than their cost i.e. $b'(m_{\tilde i}^\ast) = c_{\tilde i} \leq c_{j}$. Hence, it is optimal for agent $j$ to just use $m_{\tilde i}^\ast$ data-points, and not generate any additional data-points. Thus, the equilibrium is all other agents contribute no data, and agent $\tilde i$ computes $m_{\tilde i}^\ast$ datapoints as if on its own.

Next consider the case where $m_{\tilde i}^\ast = 0$ and hence all $m_i^\ast = 0$. Suppose all the other agents in total contribute $\Delta m$ datapoints, which is given to all agents unconditionally. With this extra free data, it is possible that there exists some agent for whom $m_{\tilde i}^\ast > 0$. However, the incentives of the agents remain identical and so the agent with the least cost collects the most data. Given access to this $\Delta m + m_{\tilde i}^\ast$ amount of data, agent $j$ with higher cost $c_j \geq c_{\tilde i}$ has no incentive to collect any additional data. Hence, only agent $\tilde i$ would collect any data and so $\Delta m =0$. However, if $\Delta m =0$, agent $\tilde i$ also has no incentive to collect any data. This implies that all agents contributing $0$ datapoints is the only Nash equilibrium possible.
\end{proof}

\section{Proofs from Section~\ref{sec:accuracy_shaping} (Accuracy Shaping under Known Costs)}

\theoremoptmechanism*

\paragraph{Proof.}
% \textbf{Data maximization.}
We will do the proof in two steps. Consider the best response $B_{\cM}$ of the agents to a mechanism $\cM$ similar to the definition in the proof of Theorem~\ref{thm:equilibrium}:
\[
    [B_{\cM}(\vm)]_i := \argmax_{m_i \geq 0} [\cM(m_i, \vm_{-i})]_i - c_i m_i \,.
\]
We will first prove that given a fixed contribution $\vm$ from other users, the best response for our mechanism $m_i^{\max}$ is higher than that of any other feasible and IR mechanism. We will then show that this necessarily implies that the equilibrium contribution of the agent is also data maximizing.

\begin{lemma}\label{lem:best-bestresponse}
For a given data contribution $\vm$ and any feasible and IR mechanism $\tilde\cM$, define best responses $B_{\cM}(\vm)$ and $B_{\tilde\cM}(\vm)$ for our mechanism $\cM$ (defined in \eqref{eqn:opt-mechanism}) and the other mechanism $\tilde \cM$. Then, for any agent $i$ and contribution $\vm$,
\[
    [B_{\cM}(\vm)]_i \geq [B_{\tilde\cM}(\vm)]_i\,.
\]
Further, the best response $[B_{\cM}(\vm)]_i$ is non-decreasing in the net contribution from other agents $(\sum_{j \neq i}\vm_j)$.
\end{lemma}
For now, we will assume that the above lemma and continue with our proof. As shown in the proof of Theorem~\ref{thm:equilibrium}, the equilibrium of all feasible mechanisms (if they exist) lie in the range $\cC := \prod_i [0, \nicefrac{1}{c_i}]$. Suppose that $\tilde\vm \in \cM$ is the equilibrium of mechanism $\tilde \cM \in \cM$. Note that $\tilde\vm$ is also the fixed point of the best response with $\tilde\vm \in B_{\tilde\cM}(\tilde\vm)$. Now, define the following subspace 
\[
    \cC_{\geq \tilde\vm} := \prod_j [\tilde \vm_j, \nicefrac{1}{c_j}]\,.
\]
The set $\cC_{\geq \tilde\vm}$ is compact and convex. Thus, we can apply Theorem~\ref{thm:equilibrium} to our optimal mechanism $\cM$ to prove that there exists an equilibrium point $\vm \in \cC_{\geq \tilde\vm}$ such that 
\[
    [\vm]_i \in \argmax_{m_i \geq \tilde m_i} [\cM(m_i, \vm_{-i})]_i - c_i m_i\,.
\]
We will next show that the above point $\vm$ is in fact a fixed of $B_{\cM}(\vm)$ and satisfies:
\[
    [\vm]_i \in \argmax_{m_i \geq 0} [\cM(m_i, \vm_{-i})]_i - c_i m_i\,.
\]
Note that the only difference between the two claims is that in the latter the $\argmax$ is taken over $\geq 0$ where as it was more constrained in the former. For the sake of contradiction, suppose this is not true i.e. there exists an agent $i$ such that $\vm_i \notin [B_{\cM}(\vm)]_i$ and $[B_{\cM}(\vm)]_i < \tilde \vm_i$. However, this leads to a contradiction:
\begin{align*}
    &\sum_{j \neq i} \vm_j \geq \sum_{j \neq i} \tilde\vm_j\\
    \Rightarrow &[B_{\cM}(\vm)]_i \geq [B_{\cM}(\tilde\vm)]_i \geq [B_{\tilde\cM}(\tilde\vm)]_i = \tilde\vm_i\,.
\end{align*}
The first inequality is because $\vm \in \cC_{\geq \tilde\vm}$. The first inequality in the second step follows from the latter part of Lemma~\ref{lem:best-bestresponse} while the next inequality is from the first part. Finally, the last equality follows because $\tilde \vm$ is a fixed point of $B_{\tilde\cM}$. Hence, we have proven that there exists a fixed point $\vm \in B_{\cM}(\vm)$ such that $\vm \in \cC_{\geq \tilde\vm}$ i.e. the equilibrium contribution of every agent under $\cM$ is at least as much as $\tilde \cM$.
\qedsymbol

\paragraph{Proof of Lemma~\ref{lem:best-bestresponse}.}
Recall the optimal mechanism defined in \eqref{eqn:opt-mechanism} restated below:
    \begin{equation}
        [\cM(\vm)]_i = 
            \begin{cases}
                a(m_i) \quad &\text{for } m_i \leq  m_i^\ast \\
                a(m_i^\ast) + (c_i + \varepsilon)(m_i - m_i^\ast)  \quad &\text{for } m_i \in [m_i^\ast, m_i^{\max}] \\
                a(\sum_j m_j)  \quad &\text{for } m_i \geq m_i^{\max} \,.
            \end{cases}
    \end{equation}
For now, suppose that $m_i^\ast > 0$. Recall, from [case 3, Theorem~\ref{thm:ind}], that this implies $a'(m_i^\ast ) = b'(m_i^\ast ) = c_i$.

First we show that $m_i^{\max}$ is the unique equilibrium contribution for an agent $i$.  The slope of the utility of agent $i$ is
\[
    u_i'(m_i; \cM) = \frac{\partial [\cM(\vm)]_i}{\partial m_i} - c_i\,.
\]
By construction, this slope is $u_i'(m_i; \cM) > 0$ for any $m_i < m_i^{\max}$. Suppose the contribution of all other agents is fixed to $\Delta m_{\text{-}i} = \sum_{j \neq i}m_j$. The slope of the utility at $m_i^{\max}$ is $u_i'(m_i^{\max}; \cM) = a'(m_i^{\max} + \Delta m_{\text{-}i}) - c_i$. Again, by construction, $m_i^{\max} + \Delta m_{\text{-}i} \geq m_i^\ast$. Since $b$ is concave and $b'$ is non-increasing,
\[
    u_i'(m_i^{\max}; \cM) = a'(m_i^{\max} + \Delta m_{\text{-}i}) - c_i = b'(m_i^{\max} + \Delta m_{\text{-}i}) - c_i \leq b'(m_i^\ast) - c_i = 0\,.
\]
Thus, $m_i^{\max}$ is the unique equilibrium contribution of agent $i$.
Next, we have to demonstrate the data-maximizing property. For the sake of contradiction, suppose there existed some other mechanism $\tilde \cM$ such that
\[
    \argmax_{m_i}[\tilde \cM(\vm)]_i - c_i m_i =: \tilde m_i > m_i^{\max}\,.
\]
This implies that $u_i'(m_i; \tilde\cM) > 0$ for any $m_i \leq \tilde m_i$,  i.e. $\frac{\partial [\tilde\cM(\vm)]_i}{\partial m_i} > c_i$. In particular, this implies that
\[
\frac{\partial [\tilde\cM(\vm)]_i}{\partial m_i} > \frac{\partial [\cM(\vm)]_i}{\partial m_i}\text{ for all } m_i \in [m_i^\ast, m_i^{\max}]\,.
\]
Further, $\tilde \cM$ satisfies individual rationality and so at $m_i = m_i^\ast$ we have 
\[
[\tilde\cM(m_i^\ast, \vm_{\text{-}i})]_i \geq a( m_i^\ast) = [\cM(m_i^\ast, \vm_{\text{-}i})]_i\,.
\]
Together, these two conditions imply that for all $m_i \in [m_i^\ast, m_i^{\max}]$, we have $[\tilde\cM(\vm)]_i > [\cM(\vm)]_i$. In particular at $m_i = m_i^{\max}$, we have
\[
[\tilde\cM(m_i^{\max}, \vm_{\text{-}i})]_i > a(\textstyle\sum_j m_j)\,.
\]
This gives us a contradiction since it violates feasibility. Thus, $m_i^{\max}$ is the maximum data which can be extracted from agent $i$.

The proofs for the low and medium cost agents are similar, while noting that $m_i^\ast = 0$. This finishes the proof of the first part. The second  part of the lemma follows directly from the definition of $\cM$ and the fact that the accuracy function $a(m_i + \sum_{j\neq i}m_j)$ is non-decreasing in the contributions $\sum_{j\neq i}m_j$. \qedsymbol

\theoremunchangedutilities*
\begin{proof}
This statement is true by construction of our mechanism. When $\e \rightarrow 0$, the slope of the utility becomes
\[
    u_i'(m_i; \tilde\cM) = \frac{\partial [\cM(\vm)]_i}{\partial m_i} - c_i = c_i + \e - c_i = 0 \text{ for all } m_i \in [m_i^\ast, m_i^{\max}]\,.
\]
Further, note that at $m_i = m_i^\ast$, we have $[\cM(m_i^\ast, \vm_{\text{-}i})]_i = a(m_i^\ast)$. Thus, for all $m_i \in [m_i^\ast, m_i^{\max}]$, the utility of agent $i$ with our mechanism $\cM$ remains constant and equal to the optimal individual utility $u_i(m_i^\ast)$.
\end{proof}

\section{Proofs from Appendix~\ref{sec:information_rent} (Data Maximization with Unverifiable Costs)}

\theoremunknownoptmechanism*
\begin{proof}
Recall that we had defined the mechanism \eqref{eqn:unknonw-opt-mechanism} as
    \begin{equation}
        [\cM(\vm)]_i = 
            \begin{cases}
                a(m_i) \quad &\text{for } m_i \leq  \high m^\ast \\
                a(\high m^\ast) + (\high c + \varepsilon)(m_i - \high m^\ast)  \quad &\text{for } m_i \in [\high m^\ast, m_i^\uparrow] \\
                a(m_i^{\downarrow} + \sum_{j\neq i} m_j)  - (\low c + \varepsilon)(m_i^\downarrow - m_i)  \quad &\text{for } m_i \in [m_i^\uparrow, m_i^\downarrow] \\
                a(\sum_j m_j)  \quad &\text{for } m_i \geq m_i^\downarrow \,.
            \end{cases}
    \end{equation}
First, we have to show that $m_i^\uparrow$ and $m_i^\downarrow$ are equilibrium for the high and low cost players $\high c$ and $\low c$ respectively. For the sake of simplicity, we first assume that $\high m^\ast > 0$ and $\low m^\ast > 0$. The proofs directly extend to the other cases. Now, note that $a'(\high m^\ast) = b'(\high m^\ast) = \high c$. Thus, by constructions, we have that for a high cost agent,
\[
    u_i'(m_i; \cM) = \frac{\partial [\cM(\vm)]_i}{\partial m_i} - \high c > 0 \text{ for all } m_i \leq m_i^\uparrow\,.
\]
where as for $m_i > m_i^\uparrow$, the slope $u_i'(m_i; \cM) = \low c + \e - \high c < 0$. Assuming $\e$ is small enough, a high cost agent obtains optimal utility at $m_i^\uparrow$. Similarly, for the low cost agent, $ u_i'(m_i; \cM)  > 0$ for all $m_i < m_i^\downarrow$ and is negative after (similar to Theorem~\ref{thm:opt-mechanism}). Thus, the optimum contribution of the low cost player is $m_i^\downarrow$. 

Next, recall that we had defined in \eqref{eqn:unknown-compromise} that $m_i^\downarrow$ satisfies%\footnote{$\clip(a, [b,c]) = \min(\max(a, b), c)$}
    \begin{equation}\label{eqn:unknown-compromise-app}
        m_i^\downarrow  = \min\rbr[\Big]{\max\rbr[\Big]{ b'^{-1}\rbr*{\low c - \tfrac{p_i }{1-p_i}\high c} - \Delta m_{\text{-}i} \ ,\  \high m^{\max}} \,,\, \low m^{\max}}\ .
    \end{equation} 
    We will show that $m_i^\downarrow$ defined this way maximizes the expected data for agent $i$:
    \begin{equation}\label{eqn:app-unknonw}
        \max_{m_i^\downarrow}\cbr*{ (1-p_i)m_i^\uparrow + p_i m_i^\downarrow } \text{ subject to }m_i^\uparrow, m_i^\downarrow  \text{ are feasible for }\cM\,.
    \end{equation}

\begin{figure}
    \centering
    \includegraphics[width=0.8\linewidth]{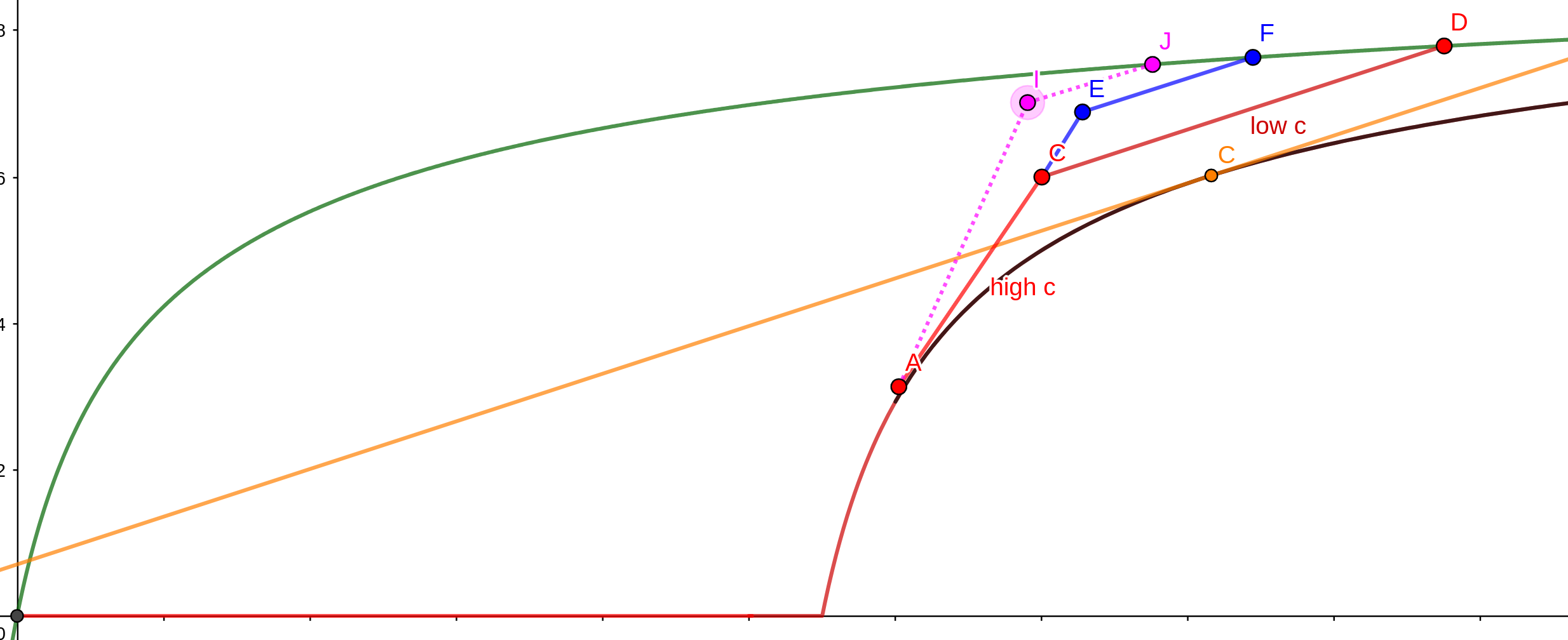}
    \caption{Accuracy shaping mechanism under unknown costs. \textit{(red curve)}: model accuracy returned to agent $i$ by the mechanism; \textit{(grey curve)}: model accuracy for agent $i$ without participation; \textit{(green curve)}: model accuracy if agent $i$ receives all the data from the other agents. Point C and D (in red) and points E and F (in blue) represent two different choices for ($m^\uparrow$ and $m^\downarrow$) respectively. If we choose a smaller value of $m^\downarrow$ (shown in blue by point F), we would see an increase in $m^\uparrow$ to point E. Thus, the optimum value balances these two depending on the probability $p_i$. Finally, points I and J (in magenta) represent potential other mechanisms $\tilde \cM$.}
    \label{fig:opt- accuracy-shaping-unknown-costs}
\end{figure}

    This involves some variational calculus (see Fig.~\ref{fig:opt- accuracy-shaping-unknown-costs}). As shown in Fig.~\ref{fig:opt- accuracy-shaping-unknown-costs}, reducing the value of $m_i^\downarrow$ results in an increase in $m_i^\uparrow$. Suppose we push the blue bar vertically by a small value $dx$. Because the slope of AC is $\high c$, this results in increase of $ \frac{dx}{\high c}$ in $m_i^\uparrow$. Correspondingly, we can show that the decrease in $m_i^\downarrow$ will be $\frac{dx}{\low c - a'(m_i^\downarrow + \Delta m_{\text{-}i})}$. Putting these together, the net expected change in data contribution is
    \[
        (1-p_i)\frac{dx}{\high c} - \frac{dx}{\low c - b'(m_i^\downarrow + \Delta m_{\text{-}i})}\,.
    \]
    The local unconstrained maxima can then be derived by setting the above to 0 i.e when
    \[
        b'(m_i^\downarrow + \Delta m_{\text{-}i}) = \low c - \tfrac{p_i }{1-p_i}\high c\,.
    \]
    Of course, we have to respect the constraints that $m_i^\downarrow \in [\high m^{\max}, \low m^{\max}]$ giving us our final result. Thus, the value of $m_i^\downarrow$ as chosen by \eqref{eqn:app-unknonw} is optimal for these class of mechanisms. 
    
    Now, we have to show that any data-maximizing mechanism corresponds to $\cM$ with some choice of $m_i^\downarrow$. Consider a mechanism $\tilde \cM$ whose equilibrium contributions are ($\tilde{m}$, $\underaccent{\tilde}{m}$) for a high and low-cost agent respectively (see points I and J in Fig.~\ref{fig:opt- accuracy-shaping-unknown-costs}). Now, from the optimality of $\low m^{\max}$, we know that $\underaccent{\tilde}{m} \leq \low m^{\max}$. Let us connect $\high m^\ast$ (point A) to $\tilde{m}$ (point I) and then to $\underaccent{\tilde}{m}$ (point J). Recall that we assumed that $\tilde \cM$ is different from $\cM$ in  \eqref{eqn:unknonw-opt-mechanism}. This means that the slope AI $\neq \high c$ or IJ $\neq \low c$. Consider the latter. Combined with I and J corresponding to equilibria, we have slope of IJ $> \low c$. This implies that starting from point I, we could have instead drawn a line segment of slope $\low c$ and increased the data contribution by the low cost agent, while keeping the contribution of the high-cost agent fixed. Similarly, we can show that the optimal slope for AI is $\high c$. Together, this implies that any optimal mechanism $\cM$ must be of the form \eqref{eqn:unknonw-opt-mechanism}, finishing our proof.
\end{proof}

\theoreminformationrent*
\begin{proof}
For a high cost player, the statement easily follows since $\frac{\partial [\cM(\vm)]_i}{\partial m_i} - \high c = 0$ for all $m_i \in [\high c^\ast, c_i^\uparrow]$. Thus, a high cost player's utility remains constant during this period and is equal to utility at $m_i = \high m^\ast$ which is $a(\high m^\ast) - \high c \high m^\ast $.

For a low cost agent, $\frac{\partial [\cM(\vm)]_i}{\partial m_i} - \high c = 0$ for all $m_i \in [m_i^\uparrow, m_i^\downarrow]$, and hence their utility is constant in this region. In particular, the difference in utility with mechanism $\cM$ and alone is
\[
a(m_i^\downarrow + \Delta m_{\text{-}i}) - \low c m_i^\downarrow - a(\low m^{\max} + \Delta m_{\text{-}i}) + \low c \low m^{\max}\,.
\]
The above quantity is always non-negative since $a'(m_i + \Delta m_{\text{-}i}) \leq \low c$ for all $m_i \in [m_i^\downarrow, \low m^{\max}]$.
\end{proof}

\end{document}